\documentclass[journal]{IEEEtran}

\usepackage{graphics,
           psfrag,
           epsfig,
           amsthm,
           cite,
           amssymb,
           url,
           dsfont,
           subfigure,
           algorithm,
           algorithmic,
           balance,
           enumerate,
           color,
           setspace
}
\usepackage{amsmath}%[centertags][tbtags][sumlimits][nosumlimits][intlimits][namelimits][nonamelimits]

\usepackage{epstopdf}

\newtheorem{definition}{Definition}

\newtheorem{theorem}{Theorem}

\newcommand{\eref}[1]{(\ref{#1})}
\newcommand{\sref}[1]{Section~\ref{#1}}
\newcommand{\appref}[1]{Appendix~\ref{#1}}
\newcommand{\fref}[1]{Figure~\ref{#1}}

\newcommand{\cref}[1]{Constraint~\ref{#1}}
\newcommand{\thref}[1]{Theorem~\ref{#1}}

\newcommand{\ignore}[1]{}

\IEEEoverridecommandlockouts
\title{Delay Minimization for Instantly Decodable Network Coding in Persistent Channels with Feedback Intermittence}

\author{
   \authorblockN{Ahmed Douik, \textit{Student Member, IEEE}, Sameh Sorour, \textit{Member, IEEE},\\ Tareq Y. Al-Naffouri, \textit{Member, IEEE}, and Mohamed-Slim Alouini, \textit{Fellow, IEEE} \vspace{-0.9cm}}
    
    {\thanks {Ahmed Douik and Mohamed-Slim Alouini are with Computer, Electrical and Mathematical Sciences and Engineering (CEMSE) Division at King Abdullah University of Science and Technology (KAUST), Thuwal, Makkah Province, Saudi Arabia, email: \{ahmed.douik,slim.alouini\}@kaust.edu.sa

Sameh Sorour is with the Electrical Engineering Department, King Fahd University of Petroleum and Minerals (KFUPM), Dhahran, Eastern Province, Saudi Arabia, email:samehsorour@kfupm.edu.sa

Tareq Y. Al-Naffouri is with both the CEMSE Division at King Abdullah University of Science and Technology (KAUST), Thuwal, Makkah Province, Saudi Arabia, and the Electrial Engineering Department at King Fahd University of Petroleum and Minerals (KFUPM), Dhahran, Eastern Province, Saudi Arabia, e-mail: tareq.alnaffouri@kaust.edu.sa.

This is an extended version of work \cite{refahmed} accepted in IEEE WiMob'13, Lyon, France, October 2013.
}}}

\begin{document}
\maketitle
\IEEEoverridecommandlockouts
\begin{abstract}
In this paper, we consider the problem of minimizing the multicast decoding delay of generalized instantly decodable network coding (G-IDNC) over persistent forward and feedback erasure channels with feedback intermittence. In such environment, the sender does not always receive acknowledgement from the receivers after each transmission. Moreover, both the forward and feedback channels are subject to persistent erasures, which can be modeled by a two state (good and bad states) Markov chain known as Gilbert-Elliott channel (GEC). Due to such feedback imperfections, the sender is unable to determine subsequent instantly decodable packets combinations for all receivers. Given this harsh channel and feedback model, we first derive expressions for the probability distributions of decoding delay increments and then employ these expressions in formulating the minimum decoding problem as a maximum weight clique problem in the G-IDNC graph. We also show that the problem formulations in simpler channel and feedback models are special cases of our generalized formulation. Since this problem is NP-hard, we design a greedy algorithm to solve it and compare it to blind approaches proposed in the literature. Through extensive simulations, our adaptive algorithm is shown to outperform the blind approaches in all situations and to achieve significant improvement in the decoding delay, especially when the channel is highly persistent.
\end{abstract}


\section{Introduction} \label{sec:intro}

In the past decade, \emph{Network Coding (NC)} emerged as a promising technique to improve throughput \cite{ref7,5753573,xiao1} and delay over wireless erasure channels \cite{1208720}. Fundamental research has been conducted to approach the network capacity. However, due to a mixing in definitions in NC, a better use of channel or throughput does not mean a lower delay, in general, at the application level \cite{1208721,49413587}. This is mainly caused by the erasure nature of links that affects the delivery of meaningful data and thus affects the ability of receivers to synchronously decode the mixed information flows. Defining the delay in network coding is not straightforward. Many recent studies have been dedicated to better understand the delay aspect of network coding. These works can be divided into two groups. The first considers the delay as the overall transmission time and the second as the individual delay experienced when delivered packets are not useful at their reception instant. The former notion is called the completion time \cite{ref4,ref12,ref17,ref13} and the latter the decoding delay \cite{ref2,ref3,ref6,ref16,ref18}. 

An important subclass of network coding, called Instantly Decodable Network Coding (IDNC) \cite{6655395,6725590,6120247,xiao2}, gained much attention thanks to its several benefits. IDNC can be implemented using simple XOR based packet encoding and decoding which eliminates the need for matrix inversion at the receiver \cite{4313060}. Moreover, each non-instantly decodable packet is discard. These simple decoding and no buffer properties allow the design of simple and power efficient receivers.

IDNC can itself be classified into two schemes named Strict (S-IDNC) and Generalized IDNC (G-IDNC). The S-IDNC scheme, proposed in \cite{ref5}, imposes restrictions on the sender to generate no packets that cannot be decoded at reception by any of the receivers. The decoding delay performance of this scheme has been studied  \cite{ref5,ref6}. However, this strict IDNC constraint limits the number of receivers targeted simultaneously. To overcome this limitation, G-IDNC was introduced in \cite{ref2}. This scheme allows the sender to generate any combination of packets but forces the receivers to discard every packet that cannot be decoded at its reception instant. Through simulations G-IDNC was shown to outperform S-IDNC and to achieve a better decoding delay \cite{ref2}.

All these aforementioned works considered an ideal G-IDNC problem with accurate and prompt feedback at the sender from all the receivers and memory-less erasure channels (MECs) for the forward links. These assumptions are too idealistic given the severe and highly correlated impairments on both forward and feedback channels of wireless networks, due to shadowing, high interference and fading. Moreover, many wireless networks standards employ a time division duplex (TDD) structure for downlink and uplink transmissions. In such networks, the sender (usually a base station or an access point) cannot receive feedback at all from any of the receivers until the start of the uplink frame. Consequently, this sender will need to transmit several subsequent packets without having any information about their reception status at the different receivers. When G-IDNC is employed in such uncertainties about prior packet reception, the sender will no longer be certain on whether a sent packet combination will instantly provide a new packet for these receivers, and thus will not be certain about the resulting completion and decoding delays. 

Some recent works started to address some of these concerns in IDNC. Especially \cite{ref12} and \cite{ref13} studied the completion time for intermittent and lossy feedback, whereas \cite{ref3} studied the decoding delay problem in persistent erasure channels. In \cite{refahmed}, we studied the problem of minimizing the decoding delay in G-IDNC over memory-less erasure forward and feedback channels with feedback intermittence. In this paper, we aim to extend our work in \cite{refahmed} by studying the minimum decoding delay problem of G-IDNC in the more general persistent erasure channel (PEC) model on both the forward and feedback channels, and in the presence of feedback intermittence. 

To achieve our goal, we first derive expressions for the probability distributions of the decoding delay increments in the different levels of feedback uncertainty over such harsh channel model. We then employ these expressions to formulate the minimum decoding delay problem as a maximum weight clique problem in the G-IDNC graph. We then show that all previously obtained results for memory-less erasure channels and/or prompt feedback can be simply obtained as special cases of our proposed generalized formulation. Since finding maximum weight cliques is NP-hard \cite{ref15,ref10,ref11,ref8}, we design a greedy algorithm to solve the problem. We finally compare the performance of our optimal and heuristic adaptive algorithms to other blind approaches proposed in \cite{ref12,ref13}.

The rest of the paper is organized as follows. \sref{sec:sysmodel} introduces the system model and presents formal definitions of the terms and parameters used in the paper. In \sref{sec:chmodel}, we illustrate the considered channel and feedback models. In \sref{sec:formulation}, we derive the probability distributions of decoding delay increments and introduce our problem formulation. \sref{sec:cases} presents the problem formulations for simpler channel and feedback environments as special cases of our proposed generalized formulation. In \sref{sec:greedy}, we design a heuristic algorithm. Simulation results are presented in \sref{sec:simulation} before concluding the paper in \sref{sec:conclusion}. 

\section{System Model and Parameters}\label{sec:sysmodel}

The model we consider in this paper consists of a wireless sender that is required to deliver different (but possibly overlapping) portions of a frame (denoted by $\mathcal{N}$) of $N$ source packets to a set (denoted by $\mathcal{M}$) of $M$ receivers. Each receiver is interested in receiving a subset $\mathcal{L}_i$ of packets from this frame (i.e. $\mathcal{L}_i \subseteq \mathcal{N},~ \forall~ i \in \mathcal{M}$). The case where $\mathcal{L}_i = \mathcal{N},~ \forall~ i \in \mathcal{M}$ is referred to as a broadcast session. Let $L$ be the average percentage of needed packets by each receive $\left(\mbox{i.e.}~L= \sum_{i \in \mathcal{M}} |\mathcal{L}_i| /(M \times N) \right)$. We refer to the packets requested by receiver $i$ (packets in $\mathcal{L}_i$) and those not requested by it (packets in $\mathcal{N}\setminus \mathcal{L}_i$) as its primary and secondary packets, respectively.

In an \emph{initial phase}, the sender transmits the $N$ packets of the frame uncoded to all the receivers. Each receiver listens to all the transmitted packets and feeds back to the sender an acknowledgement for each successfully received packet (even unwanted packets). The feedback protocol of the recovery phase in given in the next section. Both transmitted packets from the sender and feedback from the receivers are subject to erasure. Since receivers send feedback only when they successfully receive a packet, then an unheard feedback event at the sender from a given receiver makes the sender uncertain on whether the packet or the feedback were erased. At the end of this \emph{initial phase}, four sets of packets are attributed to each receiver $i$:
\begin{itemize}
\item The Has set (denoted by $\mathcal{H}_i$) is defined as the sets of packets successfully received and acknowledged by receiver $i$.
\item The Lack set (denoted by $\mathcal{L}_i$) is defined as the sets of packet that are not in the Has sets. In other words, $\mathcal{L}_i = \mathcal{N} \backslash \mathcal{H}_i$.
\item The Wants set (denoted by $\mathcal{W}_i$) is defined as the sets of primary packets in the Lack sets (i.e. primary packets that are lost by receiver $i$ or whose feedback is erased at the sender). We have $\mathcal{W}_i \subseteq \mathcal{L}_i$.
\item The Uncertain set (denoted by $\mathcal{X}_i$) is defined as the sets of  packets whose state is uncertain. We have $\mathcal{X}_i \subseteq \mathcal{L}_i$.
\end{itemize}

The sender stores this information, after transmission at time $(t-1)$, in a \emph{feedback matrix} (SFM) $\mathbf{F}(t) = [f_{ij}(t)],~ \forall~ i \in \mathcal{M},~ \forall~j \in\mathcal{N},~ \forall~t \in \mathds{N}^+$ such that:
\begin{align}
f_{ij}(t) =
\begin{cases}
0 \hspace{0.9 cm}& \text{if } j \in \mathcal{H}_i(t) \\
-1 \hspace{0.7 cm}& \text{if } j \in \mathcal{L}_i(t) \setminus \mathcal{W}_i(t) \\
1 \hspace{0.9 cm}& \text{if } j \in \mathcal{W}_i(t) \setminus \mathcal{X}_i(t) \\
x \hspace{0.9 cm}& \text{if } j \in \mathcal{W}_i(t) \cap \mathcal{X}_i(t).
\end{cases}
\end{align}

After the \emph{initial phase}, a recovery transmission phase begins at time $t=1$. In this phase, the sender transmits X-OR combination of the source packets, using information from the feedback matrix and the expected erasure patterns. After each transmission, each receiver that received and successfully decoded a packet acknowledges the reception of all received packets. The sender uses the feedback to update the feedback matrix. This process is repeated until all receivers report that they obtained all their primary packets. We define the targeted receivers by a transmission as the receivers that can instantly decode a packet from this transmission.

In the recovery phase, the transmitted coded packets can be one of the following three options for each receiver $i$:
\begin{itemize}
\item \emph{Non-innovative:} A packet is non-innovative for receiver $i$ if all the source packets it encodes were successfully received previously or they are secondary packets.
\item \emph{Instantly Decodable:} A packet is instantly decodable for receiver $i$ if it contains \emph{only one source packet} from $\mathcal{L}_i$ that was not received previously.
\item \emph{Non-Instantly Decodable:} A packet is non instantly decodable for receiver $i$ if it contains two or more source packet from $\mathcal{L}_i$.
\end{itemize}

We use the same definition of decoding delay found in \cite{ref2,ref3}:
\begin{definition}
At any recovery phase transmission, a user $i$, with non-empty Wants sets, experiences a one unit increase of decoding delay if it receives a packet that is either non-instantly decodable or both instantly decodable and non-innovative.
\end{definition}

This definition of decoding delay does not take into account delay due to erasure in transmission channels and the one due to packets arrangement and reordering. It rather focuses on the delays resulting from the reception of useless coded packets. 

\section{Channel and Feedback Models}\label{sec:chmodel}

\subsection{Forward Channel Model and Parameters}

The memory criterion of the channel is modeled by the well known Gilbert-Elliott channel (GEC) \cite{45284,4607216}. The GEC is a varying channel, the crossover probabilities of which are determined by the current state of a discrete time stationary binary Markov process (see \fref{fig:GEC}). The states are appropriately designated G for good and B for bad. Due to the underlying Markov nature of the channel, it has memory that depends on the transition probabilities between the states, which can be defined for receiver $i$ as follows.:
\begin{figure}[t]
\centering
  % Requires \usepackage{graphicx}
  \includegraphics[width=1\linewidth]{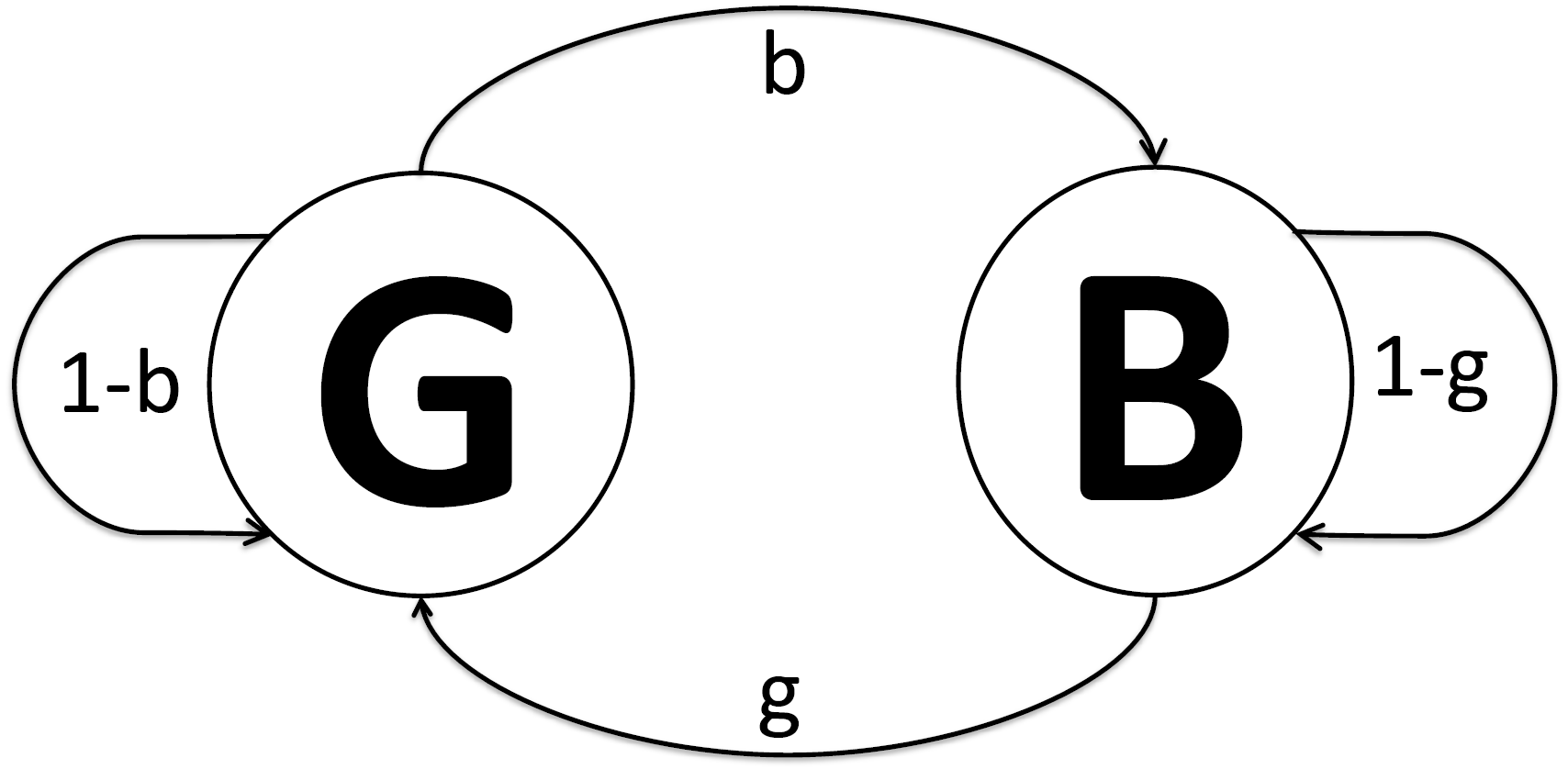}\\
  \caption{The two state Gilbert-Elliott channel.}\label{fig:GEC}
\end{figure}
\begin{align}
&\mathds{P} (C_i^p(t) = G | C_i^p(t-1) = B) = g_i^p \nonumber \\
&\mathds{P} (C_i^p(t) = B | C_i^p(t-1) = B) = 1-g_i^p \nonumber \\
&\mathds{P} (C_i^p(t) = B | C_i^p(t-1) = G) = b_i^p \nonumber \\
&\mathds{P} (C_i^p(t) = G | C_i^p(t-1) = G) = 1-b_i^p,
\end{align}
where $C_i^p(t)$ is the state of the channel for receiver $i$ at time $t ,~ \forall~ t \in \mathds{N}^+$ and the superscript $p$ refers to the forward channel. From a physical concern, the system is more likely to stay in the same state than switch between states. Thus, the values of $g_i^p$ and $b_i^p,~ \forall~ i \in \mathcal{M}$ are, in general, both less than or equal to $0.5$.

The probability to be in Good or Bad state (steady-state probabilities) can be calculated as :
\begin{align}
&\mathcal{P}_{G_i^p} = \mathds{P} (C_i^p = G) = \cfrac{g_i^p}{g_i^p+b_i^p} \nonumber \\ &\mathcal{P}_{B_i^p} = \mathds{P} (C_i^p = B) = \cfrac{b_i^p}{g_i^p+b_i^p}.
\end{align}

The memory factor of the channel is defined as $ \mu_i \triangleq 1 - g_i^p - b_i^p$ and so $0 \leq \mu_i \leq 1, ~\forall~i \in \mathcal{M}$. We also define the average channel memory $\mu = \sum_{i \in \mathcal{M}} \mu_i / M$. A high value of $\mu_i$ means that the states of the channel are highly correlated and thus the channel is likely to stay in the same state during the following transmission. When $\mu_i = 0$ (i.e. $g_i^p = 1 - b_i^p$), the state of the channel changes in an independent manner. Consequently, this channel model is a more general model of which the memory-less channel is a special case. 

We assume that the sender has perfect knowledge of the state transition probabilities of the channels of all the receivers and the initial state for each receiver. We also assume that there is no interaction between receivers and each one is seen through a channel that is independent from all the others receivers. 

\subsection{Feedback Channel Model and Parameters}

Let the time be divided into frames of length $T_f$ time-slots. Each frame is composed of a downlink sub-frame of length $T_d$ and an uplink sub-frame of length $T_u$ ($T_f = T_d+T_u$). Let $n^+(t) = \left\lceil \cfrac{t}{T_f} \right\rceil$ be the index of the frame at time $t$, where $ \lceil.\rceil$ is the ceiling function and $n^-(t) = \left\lfloor \cfrac{t}{T_f} \right\rfloor $ be the index of the downlink sub-frame before $n^+(t)$ at time $t$, where $ \lfloor.\rfloor$ is the floor function. In the intermittent lossy feedback scenario, the sender is able to transmit packets only during the downlink sub-frame and he receives feedback during the uplink sub-frame. In the uplink sub-frame, the sender receives feedback from the subset of the targeted receivers in the downlink sub-frame of the current frame. 

Let $T_{u_i}$ be the time-slot in the uplink sub-frame in which receiver $i$ sends feedback (i.e. $1 \leq T_{u_i} \leq T_u ,~ \forall~ i \in \mathcal{M}$). In other words, a receiver $i$ can transmit feedback, during frame number $n$, at time $t=n*T_f-T_u+T_{u_i}=u_i(n)$.

Similar to the forward channel, the feedback is subject to persistent loss. We model the feedback channel as a GEC. Define $g_i^q$ and $b_i^q$ as the transition probabilities, $\mathcal{P}_{G_i^q}$ and $\mathcal{P}_{B_i^q}$ as the steady-state probabilities, respectively, for each receiver $i$. Let $C_i^q(t)$ be the state of the feedback channel of receiver $i$ at time slot $t$. The memory factor of the channel for each receiver $i$ is $\psi_i \triangleq 1 - g_i^q - b_i^q,~ \forall~ i \in \mathcal{M}$. The superscript $q$ refers to the feedback channel.

\begin{figure}[t]
\centering
  % Requires \usepackage{graphicx}
  \includegraphics[width=1\linewidth]{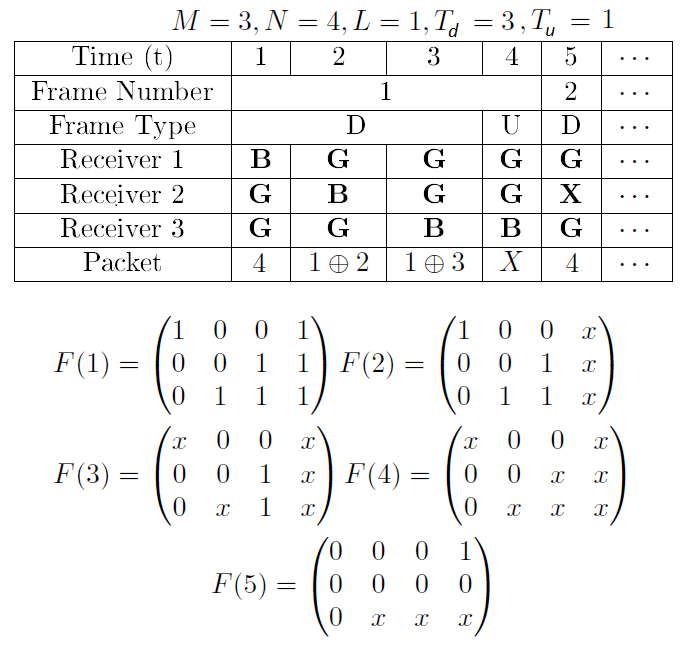}\\
  \caption{Illustration of a potential transmission in lossy intermittent feedback.}\label{fig:exemple}
\end{figure}

\fref{fig:exemple} shows an illustration of a potential transmission, in the first frame, for a lossy intermittent feedback and persistent erasure channels with all the system variables. In the example, D stands for Downlink frame, U for Uplink frame, B for Bad state and G for Good state.

We assume that for each targeted receiver there is a packet which is attempted only once from the last time a feedback is heard from that receiver. This constraint is needed to be able to estimate the state of the channel. When only one packet is left for that receiver, the state can be determined without this constraint and thus it can be removed.

We also assume, as in \cite{ref12,ref13}, that each feedback sent from a receiver includes acknowledgement of all previous received packets and that only targeted receivers will send feedback. In other words, if a feedback from one of the targeted receivers is lost, the sender will not get any feedback from this receiver until the next transmission in which it is targeted. 

The identically distributed channel can be viewed as a special case of this above model. It happens when both the transmission and the feedback channel have the same transitions probabilities (i.e. $g_i^p = g_i^q$ and $b_i^p = b_i^q, \forall~ i \in \mathcal{M}$) and the reciprocal channel, a further special case, is when both the transmission and the feedback are experiencing the same channel realization (i.e $C_i^p(t) = C_i^q(t),~ \forall~ i \in \mathcal{M}, \forall~ t \in \mathds{N}^+$).

\section{Minimum Decoding Delay Formulation}\label{sec:formulation}

In this section, we aim to derive the expected decoding delay increase for any arbitrary transmission at time $t$ and formulate the minimum decoding delay problem accordingly. To do so, we first need to derive several probability distributions in the first two subsections. 

\subsection{Transmission/Feedback Loss Probabilities at Time $t$}

In this section, we compute the probability to lose the transmission $p_i(t),~ \forall~ i \in \mathcal{M}$ and to lose the feedback $q_i(t),~ \forall~ i \in \mathcal{M}$, at time slot $t$.

In order to compute these probabilities, we first introduce the following variables: Let $n_i^{(-1)}$ and $n_i^{(0)}$ ($n_i^{(-1)} < n_i^{(0)}$) be the indices of the most recent frame where the sender heard a feedback from receiver $i$. Let $\lambda_{ij}(n)$ be the set of the time indices when packet $j$ was attempted to receiver $i$ during frame number $n$. Define $j_i^{(0)}$ as the last sent packet among those which were attempted only once between the two frames $(n_i^{(-1)}+1)$ and $n_i^{(0)}$ to receiver $i$. This variable can be mathematically defined as:
\begin{align}
j_i^{(0)} &= \underset{j \in \mathcal{W}_i(n_i^{(0)} \times T_f)}{\text{argmax}}  \bigcup_{k=n_i^{(-1)}+1}^{n_i^{(0)}} \lambda_{ij}(k) \nonumber \\
& \quad \text{subject to } \left| \bigcup_{k=n_i^{(-1)}+1}^{n_i^{(0)}} \lambda_{ij}(k) \right| = 1 ,
\end{align}
where $\cup_{x \in X}A_x $ is the union of the sets $A_x,~ \forall~ x \in X$. Let $t_i^{(0)}$ be the time where packet $j_i^{(0)}$ was attempted to receiver $i$ and $t_i^*$ the time just after the downlink frame $n_i^{(0)}$. In other words:
\begin{align}
t_i^{(0)} &= \bigcup_{k=n_i^{(-1)}+1}^{n_i^{(0)}} \lambda_{ij_i^{(0)}}(k) \\
t_i^* &= n_i^{(0)} \times T_f - T_u + 1.
\end{align}
Given these definitions, we can introduce the following theorem regarding  the loss probabilities of the forward and feedback transmissions at any given time $t$.
\begin{theorem}
The probabilities $p_i(t)$ and $q_i(t)$ of loosing a transmission from receiver $i$ or a feedback from it at time $t>t_i^*$ can be, respectively, expressed as:
\begin{align}
p_i(t) &= 
\begin{cases}
 \sum_{l=0}^{t-t_i^{(0)}-1}  \mu^l_ib_i^p \hspace{1.5 cm}&\text{if } f_{ij_i^{(0)}}(t_i^*) = 0\\
1- \sum_{l=0}^{t-t_i^{(0)}-1}  \mu^l_ig_i^p \hspace{0.3 cm}&\text{if } f_{ij_i^{(0)}}(t_i^*) = 1
\end{cases}  \\
q_i(t) &= \sum_{l=0}^{t-t_i^*+T_{u_i}-1} \psi^l_ib_i^q.
\end{align}
\label{th1}
\end{theorem}

\begin{proof}
The proof can be found in \appref{ap2}.
\end{proof}

\subsection{Decoding Delay Increment Probabilities}

Let $\mathfrak{J}=j_{k_1} \oplus j_{k_2} \oplus \cdots \oplus j_{k_n}$ be a combination of $n$ packets from $\mathcal{N}$ with $1  \leq k_i \leq N $ and $k_i \neq k_j, \forall~ i \neq j$. Let $d_i(\mathfrak{J},t)$ be the decoding delay increase experienced by receiver $i$ after the transmission of $\mathfrak{J}$ at time $t$. 

From our definition of decoding delay in \sref{sec:sysmodel}, receiver $i$ with a non-empty Wants set will not experience a decoding delay (i.e. $d_i(\mathfrak{J},t) = 0$) if and only if $\mathfrak{J}_i$ is both instantly decodable for a receiver $i$ and contains a single source packet from $\mathcal{W}_i$ (packet is innovative for receiver $i$). Define $\mathfrak{J}_i$ as the targeted packet for receiver $i$ in the transmission $\mathfrak{J}$ ( $\mathfrak{J}_i= \varnothing$ if receiver $i$ is not targeted in this transmission). Since receivers with empty Wants sets will never experience a decoding delay, we will only consider, in the rest of the section, receivers with non-empty Wants sets. Let $M_w$ be the set of these receivers.

The decoding delay increase depends not only on the sent packet $\mathfrak{J}$ but also on the channel states of the targeted receivers. We define $\tau$ as the set of targeted receivers by the transmission $\mathfrak{J}$ where the intended packet $\mathfrak{J}_i$ is a primary packet for receiver $i,~ \forall~ i \in \tau$. Let $\widehat{\tau} = M_w \setminus \tau$ be the set of non-targeted receivers and those targeted by a secondary packet.

For each receiver $i$, its Wants set is assumed to be only one of the following three options:
\begin{enumerate}
\item Non-Uncertain Wants set (i.e. $\mathcal{W}_i(t) \cap \mathcal{X}_i(t) = \varnothing$ ).
\item Partially Uncertain Wants set (i.e. $\mathcal{W}_i(t) \cap \mathcal{X}_i(t) \neq \varnothing$ and $\mathcal{W}_i(t) \not\subset \mathcal{X}_i(t)$ ).
\item Fully Uncertain Wants set (i.e. $\mathcal{W}_i(t) \subset \mathcal{X}_i(t)$).
\end{enumerate}

Let $F$ be the set of receivers having fully uncertain Wants sets and $U$ the set of targeted receivers for which intended packet state in the transmission $\mathfrak{J}$ is unknown (i.e. $f_{i\mathfrak{J}_i}(t)=x,~ \forall~ i \in U $). The following theorem gives the expected decoding delay for a receiver $i$ with non-empty Wants set.
\begin{theorem}
The probability that receiver $i$ with non-empty Wants set experiences a decoding delay at time $t$, after the transmission  $\mathfrak{J}$ is:
\begin{align}
&\mathds{P}(d_i(\mathfrak{J},t) = 1) \nonumber \\
&= 
\begin{cases}
1-p_i(t)&  i \in (\widehat{\tau} \cap \overline{F}) \\
(1-p_i(t))(1-p_{i,f}(t))&i \in (\widehat{\tau} \cap F) \\
0&  i \in (\tau \cap \overline{U}) \\
(1-p_i(t))(1-p_{i,n}(\mathfrak{J}_i,t))&  i \in (\tau \cap (U \setminus F)) \\
(1-p_i(t)) \\
\qquad \times (1-p_{i,n}(\mathfrak{J}_i,t)-p_{i,f}(t)) &i \in (\tau \cap F) ,
\end{cases}
\end{align}
where $p_{i,n}(j,t)$ (referred to as innovative probability) is the probability that packet $j$ is innovative for receiver $i$ at time $t$ and $p_{i,f}(t)$ (referred to as finish probability) is the probability that receiver $i$ successfully received all its primary packets but $\mathcal{W}_i(t) \neq \varnothing$ at the sender due to intermittence and loss of feedback.
\label{th2}
\end{theorem}

\begin{proof}
The proof can be found in \appref{ap3}.
\end{proof}

In order to derive the expressions of the innovative and finish probabilities, we introduce the following variables: Let $\mathcal{K}_{ij}$ be the set of indices of the frames in which packet $j$ was attempted to receiver $i$ since the last time the sender received feedback from this receiver, excluding the current frame. Define $\mathcal{X}_i^d(n)$ as the set of times where a packet was attempted to receiver $i$ during frame number $n$. In other words:
\begin{align}
\mathcal{X}_i^d(n) = \bigcup_{j \in \mathcal{W}_i(n \times T_f)}\lambda_{ij}(n),~ \forall~ n \in \mathds{N}^+.
\end{align}

Given these definitions, we can introduce the following two theorems, illustrating the expressions for the innovation and finish probabilities.
\begin{theorem}
The probability that packet $j$ is innovative for receiver $i$, at time $t$ is:
\begin{align}
&p_{i,n}(j,t) = \langle \overline{\prod_{k \in \lambda_{ij}(n^+(t))}} p_i(k)   \nonumber\\
& \times \overline{\prod_{k \in \mathcal{K}_{ij}}}   \left\{ \left( \prod_{s \in \mathcal{X}_i^d(k)} p_i(s)  +  \prod_{s \in \lambda_{ij}(k)} p_i(s)  \right. \right. \nonumber\\
& \left. {} \left. {} \times (1- \prod_{s \in \mathcal{X}_i^d(k) \setminus \lambda_{ij}(k)} p_i(s) )q_i(u_i(k)) \right) \right. \nonumber\\
& \left. {} \times \left( \prod_{s \in \mathcal{X}_i^d(k)} p_i(s)  + (1- \prod_{s \in \mathcal{X}_i^d(k)} p_i(s) )q_i(u_i(k)) \right)^{-1} \right\} \nonumber\\
&  \qquad \ \forall~ j \in  \mathcal{W}_i(t),~ \forall~ i \in \tau , \forall~ t \in \mathds{N}^+,
\label{innovative}
\end{align}
where 
\begin{align}
\overline{\prod_{x \in X}} (.)   = 
\begin{cases}
\prod\limits_{x \in X} (.) &\text{if } X \neq \varnothing \\
1  &\text{if } X = \varnothing .
\end{cases} 
\end{align} 
\label{th3}
\end{theorem}

\begin{proof}
The proof can be found in \appref{ap4}.
\end{proof}

\begin{theorem}
The probability that receiver $i$ successfully received all his primary packets but $\mathcal{W}_i(t) \neq \varnothing$ at time $t$ is:
\begin{align}
p_{i,f}(t)  &= \prod_{j \in \mathcal{W}_i(t)} \left( 1- p_{i,n}(j,t) \right) , ~ \forall~ i \in \tau,~ \forall~ t \in \mathds{N}^+.
\end{align}
\label{th4}
\end{theorem}

\begin{proof}
The proof can be found in \appref{ap5}.
\end{proof}

\subsection{Minimum Decoding Delay Formulation}

In order to minimize the decoding delay in G-IDNC, we look for all possible packets combinations that are instantly decodable for a subset (possibly all) of the receivers, and choose the one that has minimum expected increase in the decoding delay. To represent all these XOR packet combination, we use the G-IDNC graph proposed in \cite{ref2}. Let $\mathcal{G}(\mathcal{V},\mathcal{E})$ be the G-IDNC graph and let $v_{ij}$ be the vertices such that $\mathcal{V} = \{ v_{ij}, ~ \forall ~ j \in \mathcal{L}_i(t) ~ \mbox{and} ~ \forall i \in \mathcal{M}\}$. In other words, the vertex set of the G-IDNC graph includes a vertex $v_{ij}$ for every non-zero entry $f_{ij}(t)$ in the SFM.

Two edges $v_{ij}$ and $v_{kl}$ in $\mathcal{V}$ are connected if one of these scenarios occur:
\begin{itemize}
\item $j=l$: The receivers $i$ and $k$ are requesting the same packet $j$.
\item $j \in \mathcal{H}_k(t)$ and $l \in \mathcal{H}_i(t)$: The needed packet of each vertex is in the Has set of the receiver represented by the other vertex.
\end{itemize}

According to \cite{arg1}, minimizing the decoding delay in G-IDNC is equivalent to solving the maximum weight clique problem in the corresponding G-IDNC graph. Let $\kappa(t)$ be a maximal clique in the G-IDNC graph chosen for the transmission at time $t$. We define the packet combination $\mathfrak{J}[\kappa(t)]$ as the XOR of the union of $j$'s such that $v_{ij}\in\kappa(t)$.

The following theorem gives the minimum decoding problem formulation for G-IDNC graph in lossy intermittent feedback scenario.

\begin{theorem}
The minimum decoding delay problem for G-IDNC in lossy intermittent feedback scenario can be formulated as:
\begin{align}
\kappa^*(t) &=  \underset{\kappa(t) \in \mathcal{G}}{\text{argmax}}  \sum_{i \in \tau(\kappa(t))} (1-p_i(t)) \times \left(p_{i,n}(\mathfrak{J}_i(\kappa(t)),t) \right),
\end{align}
where $\tau(\kappa(t))$ is the set of receivers targeted by a primary packet and $\mathfrak{J}_i(\kappa(t))$ is the targeted packet in the transmission $\kappa(t)$ to receiver $i$.
\label{th5}
\end{theorem}

\begin{proof}
The proof can be found in \appref{ap6}.
\end{proof}
In other words, the minimum decoding delay problem can be formulated as a maximum weight clique problem where the weight of vertex $v_{ij}$ can be expressed as:
\begin{align}
w_{ij}^0(t) = (1-p_i(t))p_{i,n}(j,t).
\label{w0}
\end{align}

\section{Special Cases}\label{sec:cases}

In this section, we use our proposed generalized formulation to extract the distributions and problem formulations for the following seven special cases of our more general problem.

\subsection{Forward and Feedback PECs without Feedback Intermittence} 

In the forward and feedback PECs without Feedback intermittence case, the length of the frame $T_f = 2$. After each sent packet a feedback can be heard upon successful reception (i.e. $T_d=T_u=T_{u_i}=1$). The following theorem gives the simplified expression of the packet erasure and innovative probability in lossy feedback without intermittence.

\begin{theorem}
The packet erasure probability in the probabilistic feedback scenario is:
\begin{align}
p_i(t) &= \sum_{l=0}^{t-t_i^{(0)}-1} \mu_i^lb_i^p .
\end{align}
The probability that packet $j$ is innovative to receiver $i$ at time $t$ in the probabilistic feedback scenario is:
\begin{align}
p_{i,n}(j,t) &= \overline{\prod_{k \in \mathcal{K}_{ij}}}    \cfrac{p_i(2k-1)}{p_i(2k-1)+(1-p_i(2k-1))q_i(2k)} \nonumber \\
&~ \forall~ j \in \mathcal{W}_i(t),~ \forall~ i \in \tau ,~ \forall~ t \in \mathds{N}^+.
\label{sp1}
\end{align}
\label{th6}
\end{theorem}

\begin{proof}
The proof can be found in \appref{ap7}.
\end{proof}

\subsection{Forward Only PECs with Feedback Intermittence}

In the forward only PECs with feedback intermittence case, the reception of the feedback is considered perfect. In other words $q_i(t)=0, ~ \forall~ i \in \mathcal{M}, ~\forall~ t \in \mathds{N}^+$. By substituting in Equation \eref{innovative}, the probability that packet $j$ is innovative to receiver $i$ at time $t$ becomes:
\begin{align}
p_{i,n}(j,t) &= \overline{\prod_{k \in \lambda_{ij}(n^+(t))}} p_i(k)   \nonumber \\ 
&~ \forall~ j \in \mathcal{W}_i(t),~ \forall~ i \in \tau ,~ \forall~ t \in \mathds{N}^+.
\label{sp2}
\end{align}

\subsection{Forward Only PECs with Perfect Feedback}

In the forward only PECs with perfect feedback case, the length of the frame $T_f = 2$ and the reception of the feedback is considered perfect. In other words, only the data packets are subjected to loss. This case can be derived by taking the probability of loosing the feedback ($q_i(t)=0,~\forall~i \in \mathcal{M},~\forall~t \in \mathds{N}^+$) in the Equation \eref{sp1} equal to zero. Thus the probability of a packet to be innovative is $1$.
The problem can be formulated as a maximum weight clique problem as follows:
\begin{align}
\kappa^*(t) &=  \underset{\kappa(t) \in \mathcal{G}}{\text{argmax}}  \sum_{i \in \tau(\kappa(t))} (1-p_i(t)),
\label{sp3}
\end{align}
in agreement with the expression derived in \cite{ref3}.

\subsection{Forward and Feedback MECs with Feedback Intermittence}

The memory-less channel occurs when the channel changes state in a completely uncorrelated manner. This can be done in our model by setting the channel memory $\mu_i=0,~ \forall~ i \in \mathcal{M}$. Since $\mu=1-g_i-b_i=0$ then we have $g_i+b_i=1$. The transition probabilities become then independent of time:
\begin{align}
&\mathds{P} (C_i^p(t) = G | C_i^p(t-1) = B) = g_i^p &= \mathcal{P}_{G_i^p} \nonumber \\
&\mathds{P} (C_i^p(t) = B | C_i^p(t-1) = B) = 1-g_i^p &= \mathcal{P}_{B_i^p} \nonumber \\
&\mathds{P} (C_i^p(t) = B | C_i^p(t-1) = G) = b_i^p &= \mathcal{P}_{B_i^p} \nonumber \\
&\mathds{P} (C_i^p(t) = G | C_i^p(t-1) = G) = 1-b_i^p &= \mathcal{P}_{G_i^p}.
\end{align}
As a result, the packet and feedback erasure probability becomes:
\begin{align}
p_i(t) = b_i^p = p_i  ,~ \forall~ t \in \mathds{N}^+ \\
q_i(t) = b_i^q = q_i ,~ \forall~ t \in \mathds{N}^+.
\end{align}
In the forward and feedback MECs with feedback intermittence case, the probability that packet $j$ is innovative for receiver $i$ in the lossy intermittent feedback, at time $t$ becomes:
\begin{align}
&p_{i,n}(j) = p_i^{|\lambda_{ij}(n^+(t))|} \nonumber\\
& \times \overline{\prod_{k \in \mathcal{K}_{ij}}}   \cfrac{ \left(  p_i^{|\mathcal{X}_i^d(k)|}  +  p_i^{| \lambda_{ij}(k)|} \right.  \left. {} \times (1- p_i^{|\mathcal{X}_i^d(k)|-| \lambda_{ij}(k)|})q_i \right) } { \left(p_i^{|\mathcal{X}_i^d(k)|}  + (1- p_i^{|\mathcal{X}_i^d(k)|})q_i \right) }\nonumber\\
& \qquad \qquad ~ \forall~ j \in  \mathcal{W}_i,~ \forall~ i \in \tau ,
\end{align}
which is the same expression derived in \cite{refahmed}.

\subsection{Forward and Feedback  MECs without Feedback Intermittence}

In the forward and feedback  MECs without feedback intermittence case, the length of the frame $T_f = 2$. Since the transmission and feedback loss probabilities become independent for the time in the MECs then this special case can be obtained by substitution $p_i(t)=p_i$ and $q_i(t)=q_i,~\forall~t \in \mathds{N}^+$ in the Equation \eref{sp1}. As a result, the probability that packet $j$ is innovative to receiver $i$ at time $t$ in the lossy feedback without intermittence scenario becomes:
\begin{align}
p_{i,n}(j) &= \left(\cfrac{p_i}{p_i+(1-p_i)q_i}\right)^{|\mathcal{K}_{ij}|} \nonumber \\
&~ \forall~ j \in \mathcal{W}_i,~ \forall~ i \in \tau.
\end{align}
in agreement with the expression derived in \cite{refsameh}.

\subsection{Forward Only MECs with Feedback Intermittence}

In the forward only MECs with feedback intermittence case, the feedback loss probability is equal to zero. Since the transmission  loss probabilities become independent for the time in the MECs then this special case can be obtained by substitution $p_i(t)=p_i,~\forall~t \in \mathds{N}^+$ in the Equation \eref{sp2}. Thus, the probability that packet $j$ is innovative to receiver $i$ at time $t$ in the intermittent feedback scenario becomes:
\begin{align}
p_{i,n}(j) &=  p_i^{|\lambda_{ij}|} ,~ \forall~ j \in \mathcal{W}_i,~ \forall~ i \in \tau .
\end{align}
in agreement with the expression derived in \cite{refsameh}.

\subsection{Forward Only MECs with Perfect Feedback}

In the forward only MECs with perfect feedback case, the length of the frame $T_f = 2$ and the reception of the feedback is considered perfect. As for the PECs, the probability for a packet to be innovative, in the forward only MECs with perfect feedback is equal to $1$. The formulation of the maximum weight clique problem in the G-IDNC graph, in this special case, can be obtained by substituting $p_i(t)=p_i,~\forall~t \in \mathds{N}^+$ in the Equation \eref{sp3}. Thus the expression becomes:
\begin{align}
\kappa^* &=  \underset{\kappa \in \mathcal{G}}{\text{argmax}}  \sum_{i \in \tau(\kappa)} (1-p_i) ,
\end{align}
in agreement with the expression derived in \cite{ref2}.

\section{Proposed Greedy Algorithm}\label{sec:greedy}

The maximum weight clique problem in a general graph is well known to be NP-hard \cite{ref15,ref10,ref11}. Since the G-IDNC graph has a special structure, this NP-hardness result is not directly applicable. However, the author in \cite{arg1} have shown that the problem of minimizing the decoding delay for S-IDNC is equivalent to an Integer Quadratic Programming (IQP) problem which is indeed NP-hard. Since S-IDNC is a special case of G-IDNC, the previous result can be extended to G-IDNC.

To overcome this complexity problem, we design in this section a simple heuristic algorithm to solve the problem with $O(M^2N)$ complexity. (i.e. quadratic time in the number of receivers and linear time in the number of packets). This algorithm follows the same concept as the one proposed in \cite{refahmed}, but the way the new weights are computed is different as follows.

To define the news weighs, we first define $A=[a_{ij,kl}]$ as the adjacency matrix associated with the G-IDNC graph $\mathcal{G}(\mathcal{V},\mathcal{E})$ defined as follows:
\begin{align}
a_{ij,kl} = 
\begin{cases}
1 \hspace{0.7 cm} \textit{v}_{ij} \text{ is connected to } \textit{v}_{kl} \text{ in } \mathcal{G} \\
 0 \hspace{0.7 cm} \text{otherwise.}
\end{cases}
\end{align}

Define $w_{ij}(t)$ as the modified weights that take into account the connectivity of the vertex $v_{ij}$, at time $t$, to vertices having high reception probability and lower probability to be non-innovative (lower uncertainty). We define $w_{ij}(t)$ as follows:
\begin{align}
w_{ij}(t) &=  \sum_{ ~ \forall~ \textit{v}_{kl} \in \mathcal{G}} a_{ij,kl}(1-p_k(t)) \times p_{k,n}(l,t) \times \cfrac{\delta(v_{kl})}{|\mathcal{E}|},
\label{w}
\end{align}
where $\delta(v_{ij}) = \sum\limits_{~ \forall~ \textit{v}_{kl}\in \mathcal{G}} a_{ij,kl}$ is the degree of vertex $v_{ij}$ (i.e. the number of vertices adjacent to $v_{ij}$), and $|\mathcal{E}|$ is total number of edges in the graph $\mathcal{G}(\mathcal{V},\mathcal{E})$. We define the new vertexes weight $w_{ij}^*(t)$ as follows:
\begin{align}
w_{ij}^*(t) = (w_{ij}(t)+1) \times w_{ij}^0(t) ,~ \forall~ v_{ij} \in \mathcal{G}, \forall~ t \in \mathds{N}^+.
\label{w*}
\end{align}
Consequently, a vertex will have a higher weight, when:
\begin{itemize}
\item It has a large initial weight (i.e. higher probability of reception and less uncertainty) 
\item It is adjacent to a larger number of vertices with large initial weights.
\end{itemize}

To take into account the multicast characteristic of the system, the algorithm is applied on the sub-graph $\mathcal{G}_p \subseteq \mathcal{G}$, consisting of the primary vertices of the receivers, to obtain $\kappa_p$. Then, the algorithm finds $\kappa_s$ by applying the algorithm another time on the resulting sub-graph $\mathcal{G}_s$, consisting of the secondary vertices of the receivers, which are adjacent to all the vertices of the chosen clique $\kappa_p$. The final served clique is thus $\kappa^* \leftarrow \kappa_p \cup \kappa_s$. 
\begin{algorithm}[t]
\begin{algorithmic}
\REQUIRE $\mathbf{F}$ and $p_i(t),p_{i,n}(j,t),p_{i,f}(t)$ $~\forall~ i\in\mathcal{M}$  $~ \forall~ j\in\mathcal{N}$ 
\STATE Initialize $\kappa_p,\kappa_s =\varnothing$. 
\STATE Construct $\mathcal{G}_p\left(\mathcal{V}_p,\mathcal{E}_p\right)$.
\WHILE{$\mathcal{G}_p \neq \varnothing$} 
\STATE Compute $w_{ij}^0(t), w_{ij}(t)$ and $w_{ij}^*(t)$ using \eref{w0},\eref{w} and \eref{w*}.   
\STATE Select $v_p =\underset{v_{ij}\in\mathcal{G}_p}{\text{argmax}} \left\{w_{ij}^*(t)\right\}$.  
\STATE sets $\kappa_p \leftarrow \kappa_p\cup v_p$. 
\STATE sets $\mathcal{G}_p \leftarrow \mathfrak{R}(\mathcal{G}_p,v_p)$. 
\ENDWHILE
\STATE Construct $\mathcal{G}_s\left(\mathcal{V}_s,\mathcal{E}_s\right)$.
\STATE for each $v_p \in \kappa_p$
\STATE \hspace{0.2 cm} sets $\mathcal{G}_s \leftarrow \mathfrak{R}(\mathcal{G}_s,v_p)$
\WHILE{$\mathcal{G}_s \neq \varnothing$} 
\STATE Compute $w_{ij}^0(t), w_{ij}(t$) and $w_{ij}^*(t)$ using \ref{w0},\ref{w} and \ref{w*}.  
\STATE Select $v_s = \underset{v_{ij}\in\mathcal{G}_s}{\text{argmax}} \left\{w_{ij}^*(t)\right\}$. 
\STATE sets $\kappa_s \leftarrow \kappa_s\cup v_s$. 
\STATE sets $\mathcal{G}_s \leftarrow  \mathfrak{R}(\mathcal{G}_s,v_s)$.  
\ENDWHILE
\STATE sets $\kappa^* \leftarrow \kappa_p \cup \kappa_s$. 
\end{algorithmic}
\caption{Maximum Weight Vertex Search Algorithm}
\label{algo1}
\end{algorithm}
Defining $\mathfrak{R}(\mathcal{G},v_{ij})$ as the sub-graph in $\mathcal{G}$ containing only the vertices connected to $v_{ij}$, the whole algorithm structure is illustrate in Algorithm 1.

\section{Simulation Results}\label{sec:simulation}

\begin{figure}[ht]
\centering
  % Requires \usepackage{graphicx}
  \includegraphics[width=1\linewidth]{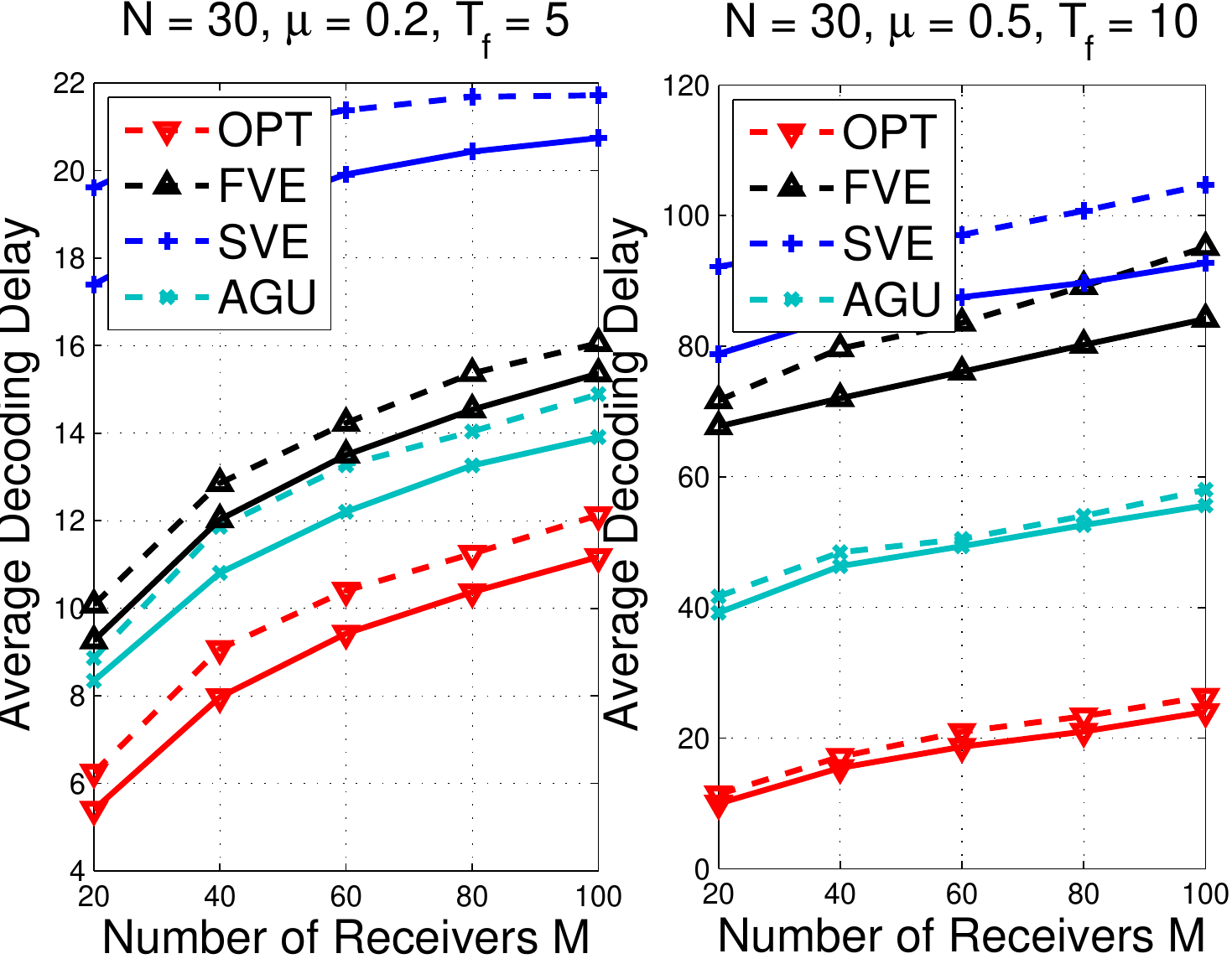}\\
  \caption{Mean decoding delays for intermittent lossy feedback versus $M$.}\label{fig:LOSSYM}
\end{figure}

In this section, we compare the performance of our adaptive algorithm (denoted by AGU), against the two partially blind graph update approaches (denoted by FVE and SVE), proposed in \cite{ref12,ref13}, and the perfect feedback (denoted by OPT), to effectively reduce the G-IDNC decoding delay in persistent erasure channels and over lossy intermittent feedback. We also compare the decoding delay performance achieved by our proposed greedy algorithm (solid line) with the heuristic proposed in \cite{refahmed} (dash line). We assume in these simulations channel reciprocity, which mean that the packet and the feedback erasure probabilities are the same (i.e. $g_i^p = g_i^q$ and $b_i^p = b_i^q ,~ \forall~ i \in \mathcal{M}$). We also assume that the packet erasure probability $b_i^p,~ \forall~ i \in \mathcal{M}$ for all receivers change uniformly in a given range $[0.1\ 0.3]$ from frame to frame while keeping its mean constant for all the simulations (even when varying the channel memory $\mu$). We compute the mean decoding delay over a large number of iterations then the average value is presented.

\begin{figure}[ht]
\centering
  % Requires \usepackage{graphicx}
  \includegraphics[width=1\linewidth]{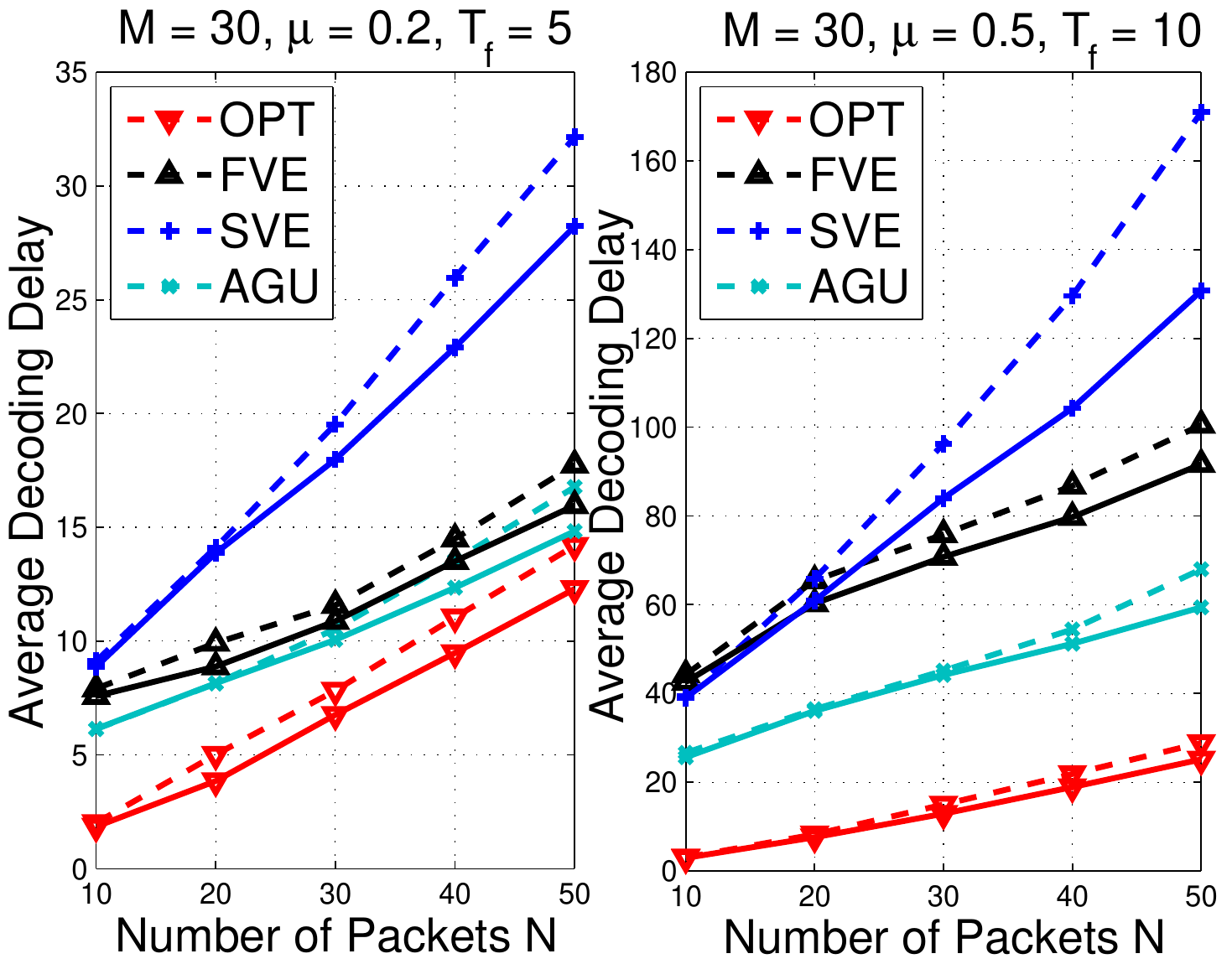}\\
  \caption{Mean decoding delays for intermittent lossy feedback versus $N$.}\label{fig:LOSSYN}
\end{figure}

\fref{fig:LOSSYM} depicts the comparison of mean decoding delays achieved by the different algorithms against $M$, for $N = 30$, $L=0.8$ and, respectively, for $\mu=0.2,T_f=5$ and $\mu=0.5,T_f=10$. \fref{fig:LOSSYN} depict the same comparison against $N$, for $M=30$, $L=0.8$ and, respectively, for $\mu=0.2,T_f=5$ and $\mu=0.5,T_f=10$. Where $L$ is the percentage of needed packets and $T_f$ the frame length. 

\begin{figure}[ht]
\centering
  % Requires \usepackage{graphicx}
  \includegraphics[width=1\linewidth]{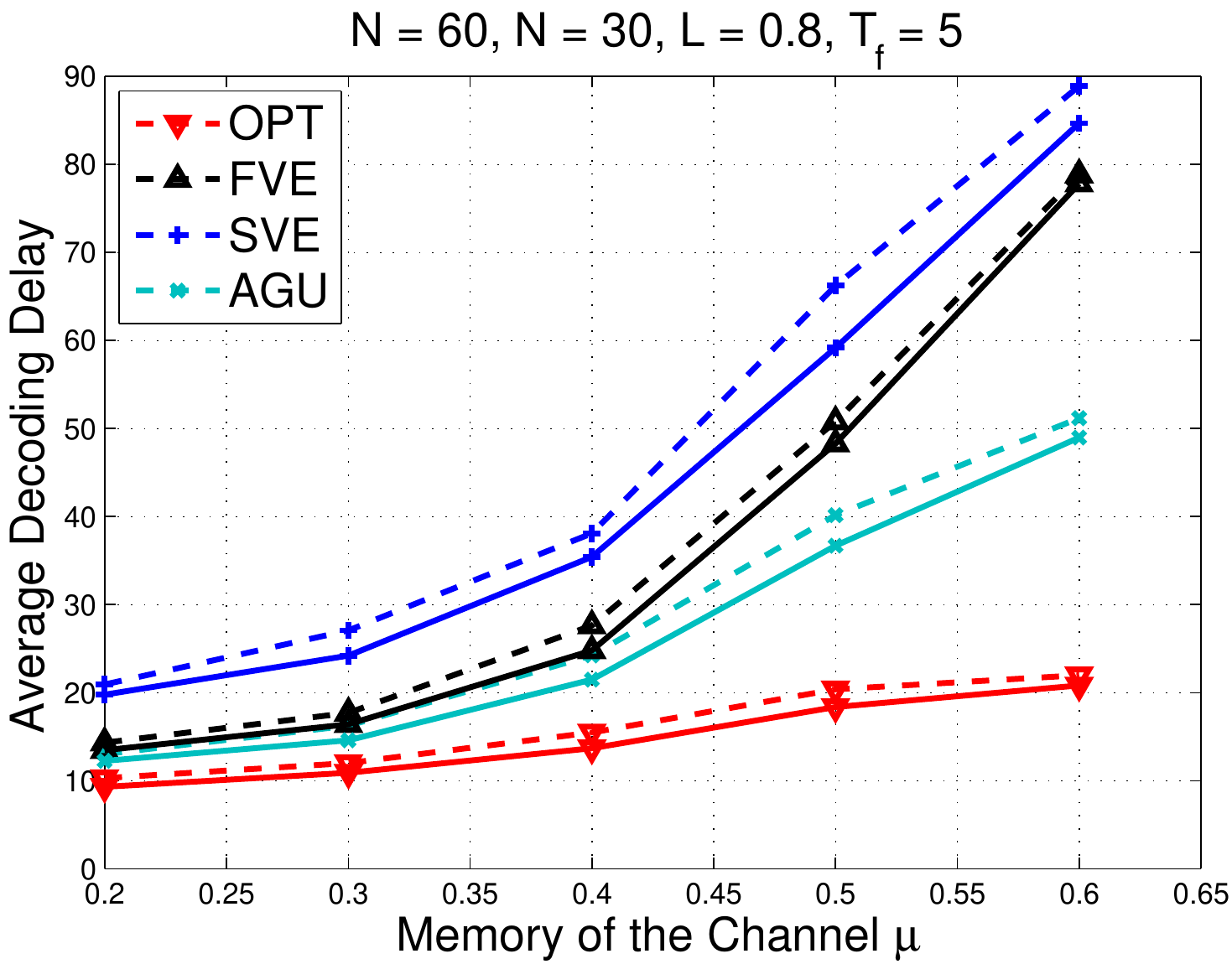}\\
  \caption{Mean decoding delays for intermittent lossy feedback versus $\mu$.}\label{fig:LOSSYMU}
\end{figure}

\begin{figure}[ht]
\centering
  % Requires \usepackage{graphicx}
  \includegraphics[width=1\linewidth]{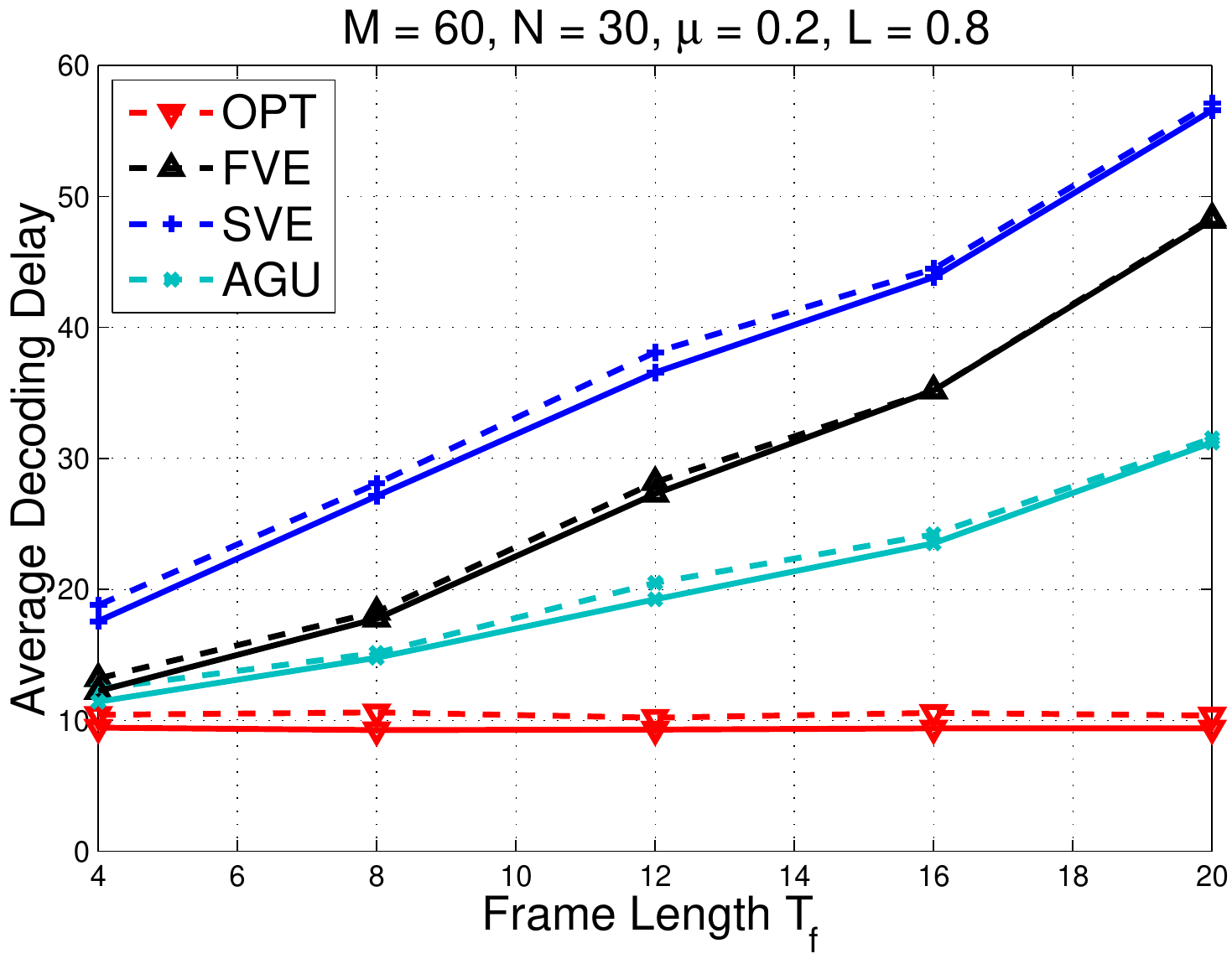}\\
  \caption{Mean decoding delays for intermittent lossy feedback versus $T_f$.}\label{fig:LOSSYTF}
\end{figure}

\begin{figure}[ht]
\centering
  % Requires \usepackage{graphicx}
  \includegraphics[width=1\linewidth]{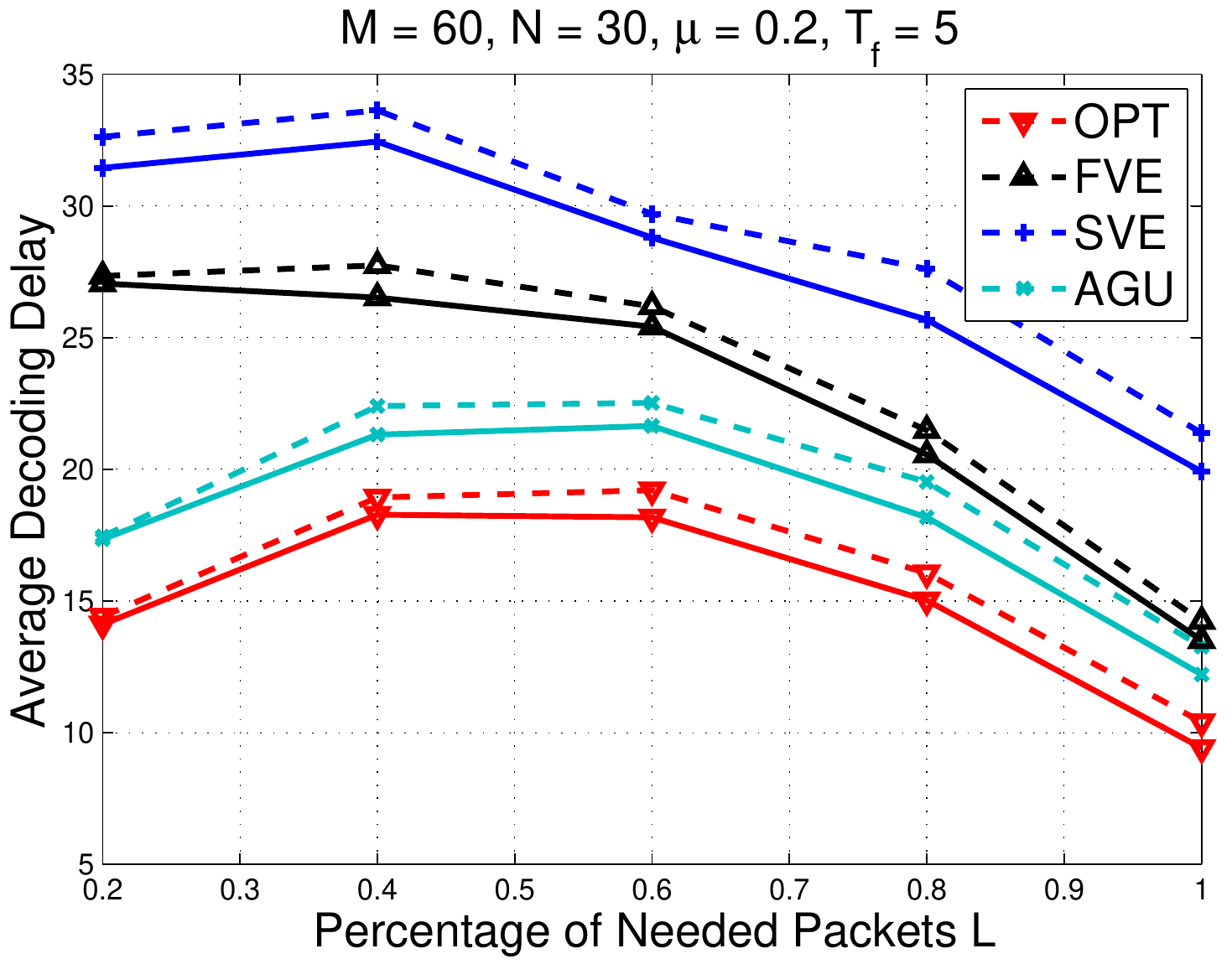}\\
  \caption{Mean decoding delays for intermittent lossy feedback versus $L$.}\label{fig:LOSSYL}
\end{figure}

\fref{fig:LOSSYMU} and \fref{fig:LOSSYTF} illustrate the comparison against $\mu$ ($T_f$) for $M=60$, $N=30$, $L=0.8$ and $T_f=5$ ($\mu=0.2$). The mean decoding delays achieved by the different algorithms against $L$ is shown in \fref{fig:LOSSYL}, for $M=60$, $N=30$, $\mu=0.2$ and $T_f=5$.

From all the figures, we can see that our proposed greedy algorithm achieve a better decoding delay in all the situation. The average decoding delay gain from \fref{fig:LOSSYM} \fref{fig:LOSSYN} and \fref{fig:LOSSYMU} is of $9\%$ comparing to the heuristic proposed in \cite{refahmed}. \fref{fig:LOSSYTF} shows that our algorithm perform a gain of $4\%$ when the percentage needed packets is less than $0.6$ and a gain higher than $9\%$ when this percentage is greater than $0.8$.

\fref{fig:LOSSYM} and \fref{fig:LOSSYN} show that all the algorithm (even the perfect feedback) achieve a close decoding delay for a low persistent channel. This can be explained by the fact that in low persistent channel, the channel change in an uncorrelated way which make channel estimation non effective for all algorithms. When the persistence of the channel becomes higher than it is more likely to remain in a state than to toggle between states, which make the estimation more accurate and explain the difference between algorithms in the decoding delay achieved.

From \fref{fig:LOSSYMU}, we clearly can see the gap between the algorithm when the persistence of the channel is higher than $0.4$ while the probability to be in the bad state still the same. We also see that our adaptive algorithm achieve reasonable degradation in high persistent channel. The same thinking is applicable for \fref{fig:LOSSYTF} our algorithm achieves a reasonable degradation in persistent channel ($\mu = 0:5$) and a large feedback period ($T_f = 10$) compared with the optimal solution (i.e perfect feedback) whereas other solutions quickly degrade.

SVE achieves the worst performance in all scenarios. This can be explained by the characteristics of this approach. In this approach, the persistent nature of the channel is not taken into account and the algorithm tries to estimate the state of uncertain packets using only the steady state probabilities. When a packet, concerning receiver $i$, is in an uncertain state, it is kept in the graph with the probability $\mathcal{P}_{B_i^p}$ and removed with the probability $\mathcal{P}_{G_i^p}$. When the channel is poorly persistent (i.e. $\mu \rightarrow 0$), then $p_i(t) \approx \mathcal{P}_{B_i^p}$ and the degradation is acceptable. However, when the memory of the channel increase, this approach completely diverge.

The same thinking is applicable for the FVE approach. In fact this approach consider all the packet received and remove all the uncertain packet from the graph until a feedback indicates the opposite. When the channel memory is low the probability $p_i(t) \approx b_i^p \ll 1$ is also low but when the persistent becomes higher the probability to lose the transmission $p_i(t) > \sum_{l=0}^{t-t_i^{(0)}-1}  \mu^lb_i^p \gg b_i^p$ becomes also higher. Thus removing all the uncertain vertexes is no longer an acceptable approach.

We clearly see that, when the percentage of packet is medium the decoding delay is high. This can be explained by the nature of the multicast scenario. In multicast, the receivers are listening to a lot of unwanted packet before completion which is reflected by a high decoding delay. When this percentage is high, the receivers are listening to less unwanted packets and thus the decoding delay is low. When the percentage is very low ($0.2$), the receivers are demanding only few packets which explain the low decoding delay.

\section{Conclusion}\label{sec:conclusion}

In this paper, the performance of generalized instantly decodable network coding in persistent erasure channels over lossy intermittent feedback to minimize the decoding delay are studied. The events than can occur at the receiver are first identified and their probabilities computed. Given these probabilities, we formulate the problem of minimizing the deciding delay and model it by a problem of maximum weight clique in the G-IDNC graph. In order to solve the former problem in linear time with the size of the graph, we design a greedy algorithm and we compared it with the heuristic proposed in \cite{refahmed}. Though extensive simulations, we show that our adaptive algorithm achieves the best decoding delay against the blind approaches proposed in \cite{ref12,ref13} for all the situations and more significant in high persistent channels. We also show that our heuristic is better adapted to reduce the decoding delay than the one proposed in \cite{refahmed} in these circumstances.

\appendices
\numberwithin{equation}{section}
\section{Auxiliary Theorem}\label{ap1}
To estimate the channel and feedback erasure probability, we first introduce the following theorem:
\begin{theorem}
Let $(X_n)_{n\geq1}$ be a two state ($x$ and $y$) Markov chain, with $P_{tr_{ x\rightarrow y}}$ and $P_{tr_{y\rightarrow x}}$ the transition probability from state $x$ to $y$ and $y$ to $x$, respectively. Let $\mu = (1-P_{tr_{ x\rightarrow y}}-P_{tr_{y\rightarrow x}})$ be the memory of the chain. Define $f(n) = \mathds{P}(X_n = y | X_{n^0} = x),~ \forall~ n \geq n^0$. We have:
\begin{align}
f(n) = P_{tr_{ x\rightarrow y}} \times \overline{\sum}_{i=0}^{n-n^0-1} \mu^{i}   
\end{align}
where 
\begin{align}
\overline{\sum}_{x \in X} (.)   = 
\begin{cases}
\sum_{x \in X} (.) &\text{if } X \neq \varnothing \\
0  &\text{if } X = \varnothing.
\end{cases} 
\end{align}
\end{theorem}

\begin{proof}
If $n=n^0$, then it is clear that:
\begin{align}
f(n) = \mathds{P}(X_{n^0} = y | X_{n^0} = x) = 0 
\end{align}
If $n>n^0$, then the relationship between $f(n)$ and $f(n-1)$ can be expressed as:
\begin{align}
&f(n) = \mathds{P}(X_n = y | X_{n^0} = x) \nonumber\\
&= \mathds{P}(X_n = y | X_{n-1} = x,X_{n^0} = x)\mathds{P}(X_{n-1} = x | X_{n^0} = x) \nonumber\\
& + \mathds{P}(X_n = y | X_{n-1} = y,X_{n^0} = x)\mathds{P}(X_{n-1} = y | X_{n^0} = x) \nonumber\\
&= \mathds{P}(X_n = y | X_{n-1} = x)\mathds{P}(X_{n-1} = x | X_{n^0} = x) \nonumber\\
& + \mathds{P}(X_n = y | X_{n-1} = y)\mathds{P}(X_{n-1} = y | X_{n^0} = x) \nonumber\\
&= (1-P_{tr_{y\rightarrow x}})f(n-1) + P_{tr_{ x\rightarrow y}}(1-f(n-1)) \nonumber\\
&= \mu \times f(n-1) + P_{tr_{ x\rightarrow y}}.
\end{align}
By a simple computation of the previous sequence, the expression of $f(n)$ becomes:
\begin{align}
f(n) &=
\begin{cases}
\mu^{n-n^0}f(n^0) + P_{tr_{ x\rightarrow y}} \times \sum_{i=0}^{n-n^0-1} \mu^{i} \hspace{0.1cm} &\text{if }  n>n^0\\
f(n^0) &\text{if }  n=n^0
\end{cases} \nonumber\\
&= \mu^{n-n^0}\times f(n^0) + P_{tr_{ x\rightarrow y}} \times \overline{\sum}_{i=0}^{n-n^0-1} \mu^{i}   \nonumber\\
&= P_{tr_{ x\rightarrow y}} \times \overline{\sum}_{i=0}^{n-n^0-1} \mu^{i}   
\end{align}
\end{proof}

\section{Proof of \thref{th1}}\label{ap2}

Note that the Wants sets do not change from the time a feedback is heard to the next time a feedback is heard from this receiver. In other words:
\begin{align}
\bigcup_{t=n_i^{(-1)}\times T_f+1}^{n_i^{(0)}\times T_f} \mathcal{W}_i(t) = \mathcal{W}_i(n_i^{(0)}\times T_f).
\end{align}
From our assumption, in \sref{sec:chmodel}, that for each targeted receiver, from the last time a feedback is heard from that receiver, there is a packet which is attempted only once, we have:
\begin{align}
~ \forall~ &i \in \mathcal{N},\exists~ j \in \mathcal{W}_i(n_i^{(0)}\times T_f) \text{ such that } \nonumber \\
& \qquad  | \bigcup_{k=n_i^{(-1)}+1}^{n_i^{(0)}} \lambda_{ij}(k)| = 1.
\end{align}
Let $\mathds{J}_i$ be the sets of packets, attempted only once, for receiver $i,~ \forall~ i \in \mathcal{M}$ from the frame $(n_i^{(-1)}+1)$ to the frame $n_i^{(0)}$. This can be expressed as follows:
\begin{align}
j \in \mathds{J}_i \Leftrightarrow    | \bigcup_{k=n_i^{(-1)}+1}^{n_i^{(0)}} \lambda_{ij}(k)| = 1.
\end{align}
It is clear that:
\begin{align}
j_i^{(0)} &= \underset{j \in \mathcal{W}_i(n_i^{(0)} \times T_f)}{\text{argmax}}  \bigcup_{k=n_i^{(-1)}+1}^{n_i^{(0)}} \lambda_{ij}(k) \nonumber \\
& \qquad \text{subject to } | \bigcup_{k=n_i^{(-1)}+1}^{n_i^{(0)}} \lambda_{ij}(k) | = 1   \nonumber \\
&= \underset{j \in \mathds{J}_i}{\text{argmax}}  \bigcup_{k=n_i^{(-1)}+1}^{n_i^{(0)}} \lambda_{ij}(k).
\end{align}
Since according to the system constraint we have $\mathds{J}_i \neq \varnothing,~ \forall~ i \in \mathcal{M}$, then the existence of $j_i^{(0)}$ is guarantee.

The state of the channel at time $t_i^{(0)}$ can be determined using the feedback received at time $t_i^*$ and depending of the reception statue of packet $j_i^{(0)}$. Then the receiver channel was good or bad according to these scenarios:
\begin{align}
C_i^p(t_i^{(0)}) = 
\begin{cases}
G \hspace{0.7 cm} \text{if } f_{i j_i^{(0)}}(t_i^*) = 0\\
B \hspace{0.7 cm} \text{if } f_{i j_i^{(0)}}(t_i^*) = 1.
\end{cases} 
\end{align}
The probability to lose the transmission or the feedback at time $t$ can then be expressed as:
\begin{align}
p_i(t) &= \mathds{P}(C_i^p(t) = B) \nonumber\\
&= 
\begin{cases}
\mathds{P}(C_i^p(t)=B|C_i^p(t_i^{(0)})= G) \hspace{0.3 cm}\text{if } f_{i j_i^{(0)}}(t_i^*) = 0\\
\mathds{P}(C_i^p(t)=B|C_i^p(t_i^{(0)})= B) \hspace{0.3 cm}\text{if } f_{i j_i^{(0)}}(t_i^*) = 1.
\end{cases}
\end{align}
According to the analysis done in Appendix A, the probability to lose the transmission at time $t$ can be expressed as:
\begin{align}
p_i(t) &= 
\begin{cases}
\overline{\sum}_{l=0}^{t-t_i^{(0)}-1}   \mu_i^lb_i^p \hspace{1.5 cm}&\text{if } f_{i j_i^{(0)}}(t_i^*) = 0\\
1 - \overline{\sum}_{l=0}^{t-t_i^{(0)}-1}   \mu_i^lg_i^p    \hspace{0.3 cm}&\text{if } f_{i j_i^{(0)}}(t_i^*) = 1.
\end{cases} 
\end{align}
Since we are computing this probability for $t\geq t_i^*$ and that $t_i^* > t_i^{(0)}$, the expression can be simplified as follows:
\begin{align}
p_i(t) &= 
\begin{cases}
 \sum_{l=0}^{t-t_i^{(0)}-1}  \mu_i^lb_i^p \hspace{1.5 cm}&\text{if } f_{i j_i^{(0)}}(t_i^*) = 0\\
1 -\sum_{l=0}^{t-t_i^{(0)}-1}  \mu_i^lg_i^p \hspace{0.3 cm}&\text{if } f_{i j_i^{(0)}}(t_i^*) = 1.
\end{cases} 
\end{align}
The state of the feedback channel at time $(t_i^*-1)$ was good, since a feedback is heard. In other words:
\begin{align}
q_i(t)  &= \mathds{P}(C_i^q(t) = B) \nonumber\\
&= \mathds{P}(C_i^q(t) = B | C_i^q(t_i^*-T_{u_i}) = G).
\end{align}
According to the analysis done in Appendix A, the probability to lose the feedback at time $t$ can be expressed as:
\begin{align}
q_i(t) &= \overline{\sum}_{l=0}^{t-t_i^*+T_{u_i}-1}   \psi_i^lb_i^q.
\end{align}
Since we are computing this probability for $t\geq t_i^*> t_i^*-1$, the expression can be simplified as follows:
\begin{align}
q_i(t) &= \sum_{l=0}^{t-t_i^*+T_{u_i}-1} \psi_i^lb_i^q.
\end{align}

\section{Proof of \thref{th2}}\label{ap3}

If receiver $i$ is not targeted in the transmission $\mathfrak{J}$ of if the intended packet for that receiver is a secondary packet then two scenarios can occur:
\begin{itemize}
\item If the Wants set of that receiver is non-uncertain or partially uncertain then it will experience one unit of delay upon successful reception. Thus, the probability of a delay increase is: 
\begin{align}
\mathds{P}( d_i(\mathfrak{J},t)=1) = 1 - p_i(t).
\end{align}
\item If the Wants set of that receiver is fully uncertain then it will experience a delay if the two following conditions are true:
\begin{enumerate}
\item $i$ received the packet
\item $i$ did not obtained all the requested packets
\end{enumerate}
Since the events are independent then the probability of a delay increase is: 
\begin{align}
\mathds{P}( d_i(\mathfrak{J},t)=1) = (1 - p_i(t))(1 - p_{i,f}(t)).
\end{align}
\end{itemize}
If receiver $i$ is targeted in the transmission $\mathfrak{J}$ by a primary packet, then three scenarios can occur:
\begin{itemize}
\item If the state of the intended packet in the SFM is known (i.e. $f_{i\mathfrak{J}_i}=1$) then this receiver will not experience a decoding delay whether the coded packet is received or not at this receiver.
\item If the state of the intended packet in the SFM is unknown (i.e. $f_{i\mathfrak{J}_i}=x$) then the expected delay increment will depend on the state of the Wants set of that receiver:
\begin{itemize}
\item[a) ] If the Wants set of that receiver is non-uncertain or partially uncertain then receiver $i$ will experience a delay if the following two events occur:
\begin{enumerate}
\item $i$ receive the packet
\item $\mathfrak{J}_i$ is not innovative for $i$
\end{enumerate}
It is clear that these events are independent then the probability of the expected delay is:
\begin{align}
\mathds{P}( d_i(\mathfrak{J},t)=1) = (1 - p_i(t))(1 - p_{i,n}(\mathfrak{J}_i,t)).
\end{align}
\item[b) ] If the Wants set of that receiver is fully uncertain then receiver $i$ will experience a delay if the following three events occur:
\begin{enumerate}
\item $i$ receive the packet
\item $\mathfrak{J}_i$ is not innovative for $i$
\item $i$ did not obtained all the requested packets
\end{enumerate}
The first event is independent of the two others. The probability of delay increment can be expressed as:
\begin{align}
\mathds{P}( d_i(\mathfrak{J},t)=1) = (1 - p_i(t)) \mathds{P}(2,3).
\end{align}
where $ \mathds{P}(2,3)$ is the probability for events 2 and 3 to happen.
Let $\mathfrak{I}_i^j$ be the event of an innovative packet $j$ for receiver $i$ and $\mathfrak{F}_i$ be the event of receiver $i$ getting all its requested packets. According to these definitions we have:
\begin{align}
\mathds{P}(2,3) &= \mathds{P}( d_i(\mathfrak{J},t)=1 | \mathfrak{I}_i^{\mathfrak{J}_i})\mathds{P}( \mathfrak{I}_i^{\mathfrak{J}_i}) \nonumber \\
& \qquad + \mathds{P}( d_i(\mathfrak{J},t)=1 | \overline{ \mathfrak{I}_i^{\mathfrak{J}_i}})\mathds{P}( \overline{\mathfrak{I}_i^{\mathfrak{J}_i}}).
\end{align}
Since it is impossible for receiver $i$ to admit a delay if the packet is innovative then:
\begin{align}
\mathds{P}(2,3) &= \mathds{P}( d_i(\mathfrak{J},t)=1 | \overline{ \mathfrak{I}_i^{\mathfrak{J}_i}}) (1 - p_{i,n}(\mathfrak{J}_i,t)).
\end{align}
We can develop the remaining term as follows:
\begin{align}
&\mathds{P}( d_i(\mathfrak{J},t)=1 | \overline{ \mathfrak{I}_i^{\mathfrak{J}_i}}) = \mathds{P}( d_i(\mathfrak{J},t) \nonumber \\
& \qquad =1 | \overline{ \mathfrak{I}_i^{\mathfrak{J}_i}},\mathfrak{F}_i)  \mathds{P}(\mathfrak{F}_i| \overline{ \mathfrak{I}_i^{\mathfrak{J}_i}}) \nonumber \\
& + \mathds{P}( d_i(\mathfrak{J},t)=1 | \overline{ \mathfrak{I}_i^{\mathfrak{J}_i}},\overline{\mathfrak{F}_i}) \mathds{P}(\overline{\mathfrak{F}_i}| \overline{ \mathfrak{I}_i^{\mathfrak{J}_i}}).
\end{align}
Since it is impossible for receiver $i$ to admit a delay if he received all the requested packets and according to the analysis done above, we have :
\begin{align}
\mathds{P}( d_i(\mathfrak{J},t)=1 | \overline{ \mathfrak{I}_i^{\mathfrak{J}_i}},\mathfrak{F}_i) =0 \\
\mathds{P}( d_i(\mathfrak{J},t)=1 | \overline{ \mathfrak{I}_i^{\mathfrak{J}_i}},\overline{\mathfrak{F}_i}) =1.
\end{align}
Then we have:
\begin{align}
\mathds{P}(2,3) &= (1 - p_{i,n}(\mathfrak{J}_i,t))\mathds{P}(\overline{\mathfrak{F}_i}| \overline{ \mathfrak{I}_i^{\mathfrak{J}_i}}) \nonumber \\
&= (1 - p_{i,n}(\mathfrak{J}_i,t)) \nonumber \\
&\qquad \times (1 - \prod_{k \in (\mathcal{W}_i \setminus \mathfrak{J}_i ) } (1-p_{i,n}(k,t)) )  \nonumber \\
&= 1-p_{i,n}(\mathfrak{J}_i,t)-p_{i,f}(t).
\end{align}
\end{itemize}
\end{itemize}
According to these scenarios, the probability that receiver $i$ with non-empty Wants set experience a decoding delay, at time $t$, after the transmission  $\mathfrak{J}$ is:
\begin{align}
&\mathds{P}(d_i(\mathfrak{J},t) = 1) \nonumber \\
&= 
\begin{cases}
1-p_i(t) \hspace{0.01 cm}& i \in (\widehat{\tau} \cap \overline{F}) \\
(1-p_i(t))(1-p_{i,f}(t)) \hspace{0.01 cm}&i \in (\widehat{\tau} \cap F) \\
0 \hspace{0.01 cm}& i \in (\tau \cap \overline{U}) \\
(1-p_i(t))(1-p_{i,n}(\mathfrak{J}_i,t)) \hspace{0.01 cm}&  i \in (\tau \cap (U \setminus F)) \\
(1-p_i(t)) \\
\qquad \times (1-p_{i,n}(\mathfrak{J}_i,t)-p_{i,f}(t)) \hspace{0.01 cm}&i \in (\tau \cap F) .
\end{cases}
\end{align}

\section{Proof of \thref{th3}}\label{ap4}

Let $\mathcal{I}_i^b(j,n^-(t)\times T_f +1)=\mathcal{I}_i^b(j,t^b)$ be the event of an innovative packet $j$ for receiver $i$ in the beginning of the frame at time $t$ (before sending any packet during that frame) and let $\mathcal{I}_i(j,t)$ be the event of that packet is innovative for that receiver at time $t$ (i.e $\mathds{P}(\mathcal{I}_i(j,t))=p_{i,n}(j,t)$). The probability that packet $j$ is innovative for receiver $i$ at time$t$ can be expressed as:
\begin{align}
p_{i,n}(j,t) &= \mathds{P}(\mathcal{I}_i(j,t)|\mathcal{I}_i^b(j,t^b))\mathds{P}(\mathcal{I}_i^b(j,t^b)) \nonumber \\
& \qquad + \mathds{P}(\mathcal{I}_i(j,t)|\overline{\mathcal{I}_i^b(j,t^b)})\mathds{P}(\overline{\mathcal{I}_i^b(j,t^b)}).
\end{align}
It is clear that if a packet is not innovative at the beginning of a frame, it cannot be innovative during that frame. Thus:
\begin{align}
\mathds{P}(\mathcal{I}_i(j,t)|\overline{\mathcal{I}_i^b(j,t^b)}) = 0.
\end{align}
Define $L_{i}^b(j,n)$ and $R_{i}^b(j,n)$ such that:
\begin{align}
L_{i}^b(j,t) &= \mathds{P}( f_{ij}(t) = 1 |  f_{ij}(t^b) = x ,\mathcal{I}_i^b(j,t^b) ) \\
R_{i}^b(j,t) &= \mathds{P}( f_{ij}(t) = 0 |  f_{ij}(t^b) = x,\mathcal{I}_i^b(j,t^b)  ) .
\end{align}
The probability for packet $j$ to be innovative to receiver $i$, at time $t$, giving that it is innovative in the beginning of the frame is:
\begin{itemize}
\item If packet $j$ was not attempted to receiver $i$ during that frame: $\mathds{P}(\mathcal{I}_i(j,t)|\mathcal{I}_i^b(j,t^b)) =  1$
\item If packet $j$ was attempted to receiver $i$ during that frame: $\mathds{P}(\mathcal{I}_i(j,t)|\mathcal{I}_i^b(j,t^b)) =  L_{i}^b(j,t)$
\end{itemize}
During the downlink frame, no feedback is expected to be received. Thus if packet $j$ is sent to receiver $i$ then:
\begin{itemize}
\item The packet is lost with probability 
\begin{align}
L_{i}^b(j,t) =  \prod\limits_{k \in \lambda_{ij}(n^+(t))} p_i(k) . 
\end{align}
\item The packet is received with probability 
\begin{align}
R_{i}^b(j,t) = (1 - \prod\limits_{k \in \lambda_{ij}(n^+(t))} p_i(k)  ).
\end{align}
\end{itemize}
Note that if packet $j$ was not attempted to receiver $i$ during the frame $n^+(t)$ then $\lambda_{ij}(n)=\varnothing$. Then the expression becomes:
\begin{align}
\mathds{P}(\mathcal{I}_i(j,t)|\mathcal{I}_i^b(j,t^b)) = \overline{\prod_{k \in \lambda_{ij}(n^+(t))}}   p_i(k) .
\end{align}
If the state of packet $j$ for receiver $i$ is known then the probability than it is innovative is $1$. If a feedback concerning packet $j$ is received from receiver $i$ at time $(n^-(t)*T_f)$ then it is clear that $\mathds{P}(\mathcal{I}_i^b(j,t^b)) = 1$.
If receiver $i$ was not targeted in the frame $n^-(t)$, then its state remains the same during all the frame. Thus 
\begin{align}
\mathds{P}(\mathcal{I}_i^b(j,t^b)) = \mathds{P}(\mathcal{I}_i^b(j,k) , ((n^-(t)-1)*T_f+1) \leq k \leq t^b.
\end{align}
Define $\mathcal{U}_i(n^+(t))$ as the unheard feedback from receiver $i$ during the frame at time $t$ and define $L_{i}(j,t)$ and $R_{i}(j,t)$ such that:
\begin{align}
L_{i}(j,t) &= \mathds{P}( f_{ij}(t) = 1 | \mathcal{U}_i(n^-(t)) ) \\
R_{i}(j,t) &= \mathds{P}( f_{ij}(t) = 0 | \mathcal{U}_i(n^-(t)) ). 
\end{align}
If receiver $i$ was targeted by packet $j$ in the frame $n^-(t)$ and no feedback is heard during uplink frame then $\mathds{P}(\mathcal{I}_i^b(j,t)) = L_{i}(j,t)$. An unheard feedback $\mathcal{U}_i(n^+(t))$ can occur if one of these scenarios occurs:
\begin{itemize}
\item All the packet sent during downlink frame number $n^+(t)$ are lost. This happen with probability 
\begin{align}
\prod\limits_{s \in \mathcal{X}_i^d(n^+(t))} p_i(s).
\end{align}
\item At least one packet arrived but the feedback is lost. This event happen with probability 
\begin{align}
(1-\prod\limits_{s \in \mathcal{X}_i^d(n^+(t))} p_i(s))q_i(u_i(n^+(t))).
\end{align}
\end{itemize}
Then the probability to event $\mathcal{U}_i(n^+(t))$ to occur is:
\begin{align}
\mathds{P}(\mathcal{U}_i(n^+(t))) & = \prod\limits_{s \in \mathcal{X}_i^d(n^+(t))} p_i(s) \nonumber\\
& \qquad + (1-\prod\limits_{s \in \mathcal{X}_i^d(n^+(t))} p_i(s))q_i(u_i(n^+(t)))).
\end{align}
When the receiver $i$ is targeted, in frame $n^-(t)$, and the feedback is not received, packet $j$ can be lost or received according to theses scenarios:
\begin{itemize}
\item The packet is lost if one of these events occurs:
\begin{itemize}
\item All the sent packets during downlink frame number $n^-(t)$ are lost. This event happen with probability $ \prod\limits_{s \in \mathcal{X}_i^d(n^-(t))} p_i(s)$.
\item The considered packet is lost, at least one of the other packet arrived and the feedback is lost. This event happen with probability $\prod\limits_{s \in \lambda_{ij}(n^-(t))} p_i(s)  
 \times (1- \prod\limits_{s \in \mathcal{X}_i^d(n^-(t)) \setminus \lambda_{ij}(n^-(t))} p_i(s) )q_i(u_i(n^-(t))) $.
\end{itemize}
\item The packet is received if the packet arrived and the feedback is lost. This event happen with probability $(1- \prod\limits_{s \in \lambda_{ij}(n^-(t))} p_i(s) )q_i(u_i(n^-(t)))$. 
\end{itemize}
Considering theses events, the expressions of $L_{i}(j,t)$ and $R_{i}(j,t)$ becomes:
\begin{align}
& L_{i}(j,t) = \left(\prod\limits_{s \in \mathcal{X}_i^d(n^-(t))} p_i(s) + \prod\limits_{s \in \lambda_{ij}(n^-(t))} p_i(s)  \right.\nonumber\\
& \left. {} \times (1- \prod\limits_{s \in \mathcal{X}_i^d(n^-(t)) \setminus \lambda_{ij}(n^-(t))} p_i(s) )q_i(u_i(n^-(t) )) \right) \nonumber \\
& \qquad \times \left(\prod\limits_{s \in \mathcal{X}_i^d(n^-(t))} p_i(s) \right.\nonumber\\
& \left. {} \qquad + (1-\prod\limits_{s \in \mathcal{X}_i^d(n^-(t))} p_i(s))q_i(u_i(n^-(t))))\right)^{-1} .
\end{align}
\begin{align}
&R_{i}(j,t) = \left( (1- \prod\limits_{s \in \lambda_{ij}(n^-(t))} p_i(s) )q_i(u_i(n^-(t) )) \right) \nonumber\\
& \qquad \times \left(\prod\limits_{s \in \mathcal{X}_i^d(n^-(t))} p_i(s) \right.\nonumber\\
& \left. {} \qquad + (1-\prod\limits_{s \in \mathcal{X}_i^d(n^-(t))} p_i(s))q_i(u_i(n^-(t))))\right)^{-1}.
\end{align}
Considering these expressions, the probability for packet $j$ to be innovative to receiver $i$ in the beginning of the frame $n^+(t)$ is:
\begin{itemize}
\item if a feedback is received in the frame $n^-(t)$: 
\begin{align}
\mathds{P}(\mathcal{I}_i^b(j,n)) = 1.
\end{align}
\item if no feedback is received in the frame $n^-(t)$ then:
\begin{itemize}
\item if receiver $i$ was not targeted in frame $n^-(t)$: $\mathds{P}(\mathcal{I}_i^b(j,t^b)) = \mathds{P}(\mathcal{I}_i^b(j,k)), ((n^-(t)-1)*T_f+1) \leq k \leq t^b$
\item if receiver $i$ was targeted in frame $n^-(t)$:
\end{itemize}
\end{itemize}
\begin{align}
\mathds{P}(\mathcal{I}_i^b(j,t^b)) &= L_{i}(j,t) \mathds{P}(\mathcal{I}_i^b(j,(n^-(t)-1)*T_f)).
\end{align}
Then, if the state of packet $j$ is uncertain for receiver $i$, the probability for that packet $j$ to be innovative in the beginning of the  frame at time $t^b$ is:
\begin{align}
\mathds{P}&(\mathcal{I}_i^b(j,t^b)) = \prod_{k \in \mathcal{K}_{ij}} \left\{ \left( \prod_{s \in \mathcal{X}_i^d(k)} p_i(s)  +  \prod_{s \in \lambda_{ij}(k)} p_i(s)  \right. \right. \nonumber\\
& \left. {} \left. {} \times (1- \prod_{s \in \mathcal{X}_i^d(k) \setminus \lambda_{ij}(k)} p_i(s) )q_i(u_i(k )) \right) \right. \nonumber\\
& \left. {} \times \left( \prod_{s \in \mathcal{X}_i^d(k)} p_i(s)  + (1- \prod_{s \in \mathcal{X}_i^d(k)} p_i(s) )q_i(u_i(k )) \right)^{-1} \right\} \nonumber\\
&  \qquad ~ \forall~ j \in  \mathcal{W}_i(t),~ \forall~ i \in \tau \cap U ,~ \forall~ t \in \mathds{N}^+.
\end{align}
Note that if the state of packet $j$ is certain for receiver $i$, then this packet has never been attempted since the last feedback (i.e. $\mathcal{K}_{ij} = \varnothing$). Then the probability for packet $j$ to be innovative to receiver $i$ in the beginning of the  frame at time $t^b$ can be expressed as:
\begin{align}
\mathds{P}&(\mathcal{I}_i^b(j,t^b)) = \overline{\prod_{k \in \mathcal{K}_{ij}}}   \left\{ \left( \prod_{s \in \mathcal{X}_i^d(k)} p_i(s)  +  \prod_{s \in \lambda_{ij}(k)} p_i(s)  \right. \right. \nonumber\\
& \left. {} \left. {} \times (1- \prod_{s \in \mathcal{X}_i^d(k) \setminus \lambda_{ij}(k)} p_i(s) )q_i(u_i(k)) \right) \right. \nonumber\\
& \left. {} \times \left( \prod_{s \in \mathcal{X}_i^d(k)} p_i(s)  + (1- \prod_{s \in \mathcal{X}_i^d(k)} p_i(s) )q_i(u_i(k )) \right)^{-1} \right\} \nonumber\\
&  \qquad ~ \forall~ j \in  \mathcal{W}_i(t),~ \forall~ i \in \tau ,~ \forall~ t \in \mathds{N}^+.
\end{align}
Considering theses expressions, the probability that packet $j$ is innovative for receiver $i$ at time $t$ is:
\begin{align}
&p_{i,n}(j,t) = \overline{\prod_{k \in \lambda_{ij}(n^+(t))}} p_i(k)   \nonumber\\
& \times \overline{\prod_{k \in \mathcal{K}_{ij}}}   \left\{ \left( \prod_{s \in \mathcal{X}_i^d(k)} p_i(s)  +  \prod_{s \in \lambda_{ij}(k)} p_i(s)  \right. \right. \nonumber\\
& \left. {} \left. {} \times (1- \prod_{s \in \mathcal{X}_i^d(k) \setminus \lambda_{ij}(k)} p_i(s) )q_i(u_i(k )) \right) \right. \nonumber\\
& \left. {} \times \left( \prod_{s \in \mathcal{X}_i^d(k)} p_i(s)  + (1- \prod_{s \in \mathcal{X}_i^d(k)} p_i(s) )q_i(u_i(k )) \right)^{-1} \right\} \nonumber\\
&  \qquad ~ \forall~ j \in  \mathcal{W}_i(t),~ \forall~ i \in \tau ,~ \forall~ t \in \mathds{N}^+.
\end{align}

\section{Proof of \thref{th4}}\label{ap5}

If receiver $i$ do not have all his primary needed packets in an uncertain state, then it is clear that:
\begin{align}
p_{i,f}(t) = 0 ,~ \forall~ i \in ( \tau \setminus F ),~ \forall~ t \in \mathds{N}^+.
\end{align} 
If receiver $i$ have all his primary packets in an uncertain state then the event of finishing can occur only if these two conditions are true:
\begin{enumerate}
\item The Uncertain sets is equal to the Wants sets ( i.e $\mathcal{W}_i(t) = \mathcal{X}_i(t)$). In other words, all the primary packets were attempted with no feedback heard from that receiver. 
\item All the packets in the Uncertain sets were successfully received by receiver $i$.
\end{enumerate}
The probability that a packet $j$ is received by receiver $i$ but not fed back before time $t$ is:
\begin{align}
\mathds{P}( &f_{ij}(t) = 0 | f_{ij}(n^-(t) \times T_f +1 ) = x )  \nonumber \\
 &= 1-\mathds{P}( f_{ij}(t) = 1 | f_{ij}(n^-(t) \times T_f +1) = x ) \nonumber \\
& = 1-p_{i,n}(j,t).
\end{align}
From the above conditions, we can express $p_{i,f}$ as:
\begin{align}
p_{i,f}(t)  &= \prod_{j \in \mathcal{X}_i(t)}(1-p_{i,n}(j,t)) \nonumber \\
&= \prod_{j \in \mathcal{W}_i(t)}(1-p_{i,n}(j,t)) ,~ \forall~ i \in \tau \cap F ,~ \forall~ t \in \mathds{N}^+.
\end{align}
If receiver $i$ do not have all his primary needed packets in an uncertain state, then $\exists j^* \in (\mathcal{W}_i(t) \setminus \mathcal{X}_i(t))$ with $p_{i,n}(j^*,t) = 1$. Thus, we have:
\begin{align}
\prod_{j \in \mathcal{W}_i(t)}(1-p_{i,n}(j,t)) = 0  ,i \in ( \tau \setminus F ),~ \forall~ t \in \mathds{N}^+.
\end{align}
The probability that receiver $i$ got all the requested packets but $\mathcal{W}_i(t) \neq \varnothing$, at time $t$ is:
\begin{align}
p_{i,f}(t)  &= \prod_{j \in \mathcal{W}_i(t)} \left( 1- p_{i,n}(j,t) \right) , ~ \forall~ i \in \tau,~ \forall~ t \in \mathds{N}^+.
\end{align}

\section{Proof of \thref{th5}}\label{ap6}

Let $\mathcal{D}\left(\kappa,t\right)$ be the sum of the decoding delay increases of all receivers after this transmission (i.e $\mathcal{D}\left(\kappa,t\right) = \sum\limits_{i=1}^{M} d_i(\mathfrak{J},t)$). According to the analysis done in \sref{sec:formulation}, the expected sum decoding delay increase after this transmission can be expressed as:
\begin{align}
\mathds{E}&\left[\mathcal{D}(\kappa,t)\right] = \sum_{ i \in \widehat{\tau}(\kappa(t)) \cap \overline{F}} (1-p_i(t)) \nonumber \\
&  + \sum_{ i \in \widehat{\tau}(\kappa(t)) \cap F} (1-p_i(t))(1-p_{i,f}(t)) \nonumber \\
&   + \sum_{ i \in \tau(\kappa(t)) \cap (U \setminus F)} (1-p_i(t))(1-p_{i,n}(\mathfrak{J}_i(\kappa(t)),t)) \nonumber \\
&  +  \sum_{ i \in \tau(\kappa(t)) \cap F} (1-p_i(t))(1-p_{i,n}(\mathfrak{J}_i(\kappa(t)),t)-p_{i,f}(t)).
\end{align}
To minimize the decoding delay, the chosen clique in the G-IDNC graph must minimize the expected mean decoding delay. The clique problem can be formulated as:
\begin{align}
&\kappa^{*} (t) =  \underset{\kappa(t) \in \mathcal{G}}{\text{argmin}}\left\{ \mathds{E} \left[ \mathcal{D}\left(\kappa,t\right) \right]\right\}  \nonumber\\ 
&=  \underset{\kappa(t) \in \mathcal{G}}{\text{argmin}} \left\{ 
\sum_{ i \in \widehat{\tau}(\kappa(t)) \cap \overline{F}} (1-p_i(t)) \right.\nonumber \\
&  \left. {}+ \sum_{ i \in \widehat{\tau}(\kappa(t)) \cap F} (1-p_i(t))(1-p_{i,f}(t)) \right.\nonumber \\
&  \left. {} + \sum_{ i \in \tau(\kappa(t)) \cap (U \setminus F)} (1-p_i(t))(1-p_{i,n}(\mathfrak{J}_i(\kappa(t)),t)) \right.\nonumber \\
&  \left. {} + \sum_{ i \in \tau(\kappa(t)) \cap F} (1-p_i(t))(1-p_{i,n}(\mathfrak{J}_i(\kappa(t)),t)-p_{i,f}(t)) \right\} \nonumber \\
&=  \underset{\kappa \in \mathcal{G}}{\text{argmax}} \left\{ \sum_{ i \in \tau(\kappa(t)) \cap \overline{F}} (1-p_i(t)) \right.\nonumber \\
&  \left. {}+ \sum_{ i \in \tau(\kappa(t)) \cap F} (1-p_i(t))(1-p_{i,f}(t)) \right.\nonumber \\
&  \left. {} - \sum_{ i \in \tau(\kappa(t)) \cap (U \setminus F)} (1-p_i(t))(1-p_{i,n}(\mathfrak{J}_i(\kappa(t)),t)) \right.\nonumber \\
&  \left. {} - \sum_{ i \in \tau(\kappa(t)) \cap F} (1-p_i(t))(1-p_{i,n}(\mathfrak{J}_i(\kappa(t)),t)-p_{i,f}(t)) \right\}.
\end{align}
It is clear that:
\begin{align}
&\sum_{ i \in \tau(\kappa(t)) \cap \overline{F}} (1-p_i(t))  + \sum_{ i \in \tau(\kappa(t)) \cap F} (1-p_i(t))(1-p_{i,f}(t)) \nonumber \\
= &\sum_{ i \in \tau(\kappa(t))} (1-p_i(t)) - \sum_{ i \in \tau(\kappa(t)) \cap F} (1-p_i(t))p_{i,f}(t) . 
\end{align}
We can develop the last term as follows:
\begin{align}
&\sum_{ i \in \tau(\kappa(t)) \cap F} (1-p_i(t))(1-p_{i,n}(\mathfrak{J}_i(\kappa(t)),t)-p_{i,f}(t)) \nonumber \\
&= \sum_{ i \in \tau(\kappa(t)) \cap F} (1-p_i(t))(1-p_{i,n}(\mathfrak{J}_i(\kappa(t)),t)) \nonumber \\
& - \sum_{ i \in \tau(\kappa(t)) \cap F} (1-p_i(t))p_{i,f}(t).
\end{align}
By simple computation the above expression can be simplified as follows:
\begin{align}
\kappa^{*} (t) &=  \underset{\kappa \in \mathcal{G}}{\text{argmax}} \left\{ \sum_{ i \in \tau(\kappa(t))} (1-p_i(t)) \right.\nonumber \\
&  \left. {} - \sum_{ i \in \tau(\kappa(t)) \cap (U \setminus F)} (1-p_i(t))(1-p_{i,n}(\mathfrak{J}_i(\kappa(t)),t)) \right.\nonumber \\
&  \left. {} - \sum_{ i \in \tau(\kappa(t)) \cap F} (1-p_i(t))(1-p_{i,n}(\mathfrak{J}_i(\kappa(t)),t)) \right\} \nonumber \\
&= \underset{\kappa \in \mathcal{G}}{\text{argmax}} \left\{ \sum_{ i \in \tau(\kappa(t))} (1-p_i(t)) \right.\nonumber \\
&  \left. {} - \sum_{ i \in \tau(\kappa(t)) \cap U } (1-p_i(t))(1-p_{i,n}(\mathfrak{J}_i(\kappa(t)),t)) \right\}.
\end{align}
Note that when the state of the targeted packet in a given transmission is known, then the probability that this packet is innovative is $1$ (i.e. $p_{i,n}(\mathfrak{J}_i(\kappa(t)),t) = 1 ,~ \forall~ i \in (\tau(\kappa) \setminus U, \forall~ t \in \mathds{N}^+$). The sum can be expressed as:
\begin{align}
&\sum_{i \in \tau(\kappa(t)) \cap U } \left( 1 - p_i(t) \right)\left(1- p_{i,n}\left(\mathfrak{J}_i(\kappa(t)),t\right) \right) \nonumber \\
& \qquad= \sum_{i \in \tau(\kappa(t)) } \left( 1 - p_i(t) \right)\left(1- p_{i,n}\left(\mathfrak{J}_i(\kappa(t)),t\right) \right). 
\end{align}
The formulation of the decoding delay becomes:
\begin{align}
\kappa^*(t) &=  \underset{\kappa \in \mathcal{G}}{\text{argmax}}  \sum_{i \in \tau(\kappa(t))} (1-p_i(t)) \times p_{i,n}(\mathfrak{J}_i(\kappa(t)),t).  
\end{align}

\section{Proof of \thref{th6}} \label{ap7}

In this configuration, the length of the downlink frame is equal to $1$, in other words $T_f=2$. Each receiver send feedback upon successful reception and since only one packet at most per frame is attempted than the system constraint that there is at least one packet attempted only once can removed. In fact it is always possible to accurately determine the state of the channel upon successful reception of the feedback. If a feedback is received at time $t$ it means that the receiver channel was good at $t-2$. The erasure probability can be simplified:
\begin{align}
p_i(t) &= \sum_{l=0}^{t-t_i^{(0)}-1} \mu_i^lb_i^p .
\end{align}
Packets are attempted in this configuration at time $t=2*k+1, ~ \forall~ k \in \mathds{N}$ and the feedback is received at time $t=2*k,~ \forall~ k \in \mathds{N}^+$. Thus, we have:
\begin{align}
p_{i,n}(j,t) &= p_{i,n}(j,2n^+(t)+1) \nonumber \\
& \qquad ~ \forall~ i \in \mathcal{M},~ \forall~ j\in \mathcal{W}_i,~ \forall~ t \in \mathds{N}^+.
\end{align}
Since we are always computing the expected decoding delay at the beginning of the frame ($t=2*k+1$), then no packet is yet attempted during this frame. In other words $\lambda_{ij}(n^+(t)) = \varnothing$, which gives:
\begin{align}
\overline{\prod_{k \in \lambda_{ij}(n^+(t))}} p_i(k)   = 1.
\end{align}
Note that only one packet can be attempted during a frame then $|\lambda_{ij}(n)| \leq 1$ and $|\mathcal{X}_i^d(n)| \leq 1, \forall~ n \in \mathds{N}^+$. More specifically, we have:
\begin{align}
\begin{cases}
\lambda_{ij}(k) = \mathcal{X}_i^d(k) = 2k-1 \hspace{0.3cm} &\text{if } k \in \mathcal{K}_{ij} \\
\lambda_{ij}(k) = \varnothing &\text{if } k \notin \mathcal{K}_{ij}.
\end{cases}
\end{align}
The probability that packet $j$ is innovative to receiver $i$ at time $t$ can be simplified as follows:
\begin{align}
p_{i,n}(j,t) &= \overline{\prod_{k \in \mathcal{K}_{ij}}}    \cfrac{p_i(2k-1)}{p_i(2k-1)+(1-p_i(2k-1))q_i(2k)} \nonumber \\
&~ \forall~ j \in \mathcal{W}_i(t),~ \forall~ i \in \tau ,~ \forall~ t \in \mathds{N}^+.
\end{align}

\bibliographystyle{IEEEtran}
\bibliography{references}

\end{document}